%% file: blocksequences.tex
\documentclass[a4paper,UKenglish,envcountsame]{llncs}
\usepackage{amsmath,amssymb}
\usepackage[utf8]{inputenc}
\usepackage{todonotes}

\usepackage[english]{babel}
\usepackage[T1]{fontenc}

\usepackage{thmtools}
\usepackage{chngcntr}

\def\N{\mathbb{N}}

\def\ta{\mathtt{a}}
\def\tb{\mathtt{b}}
\def\tc{\mathtt{c}}
\def\td{\mathtt{d}}
\def\tn{\mathtt{n}}

\DeclareMathOperator{\Fact}{Fact}

\DeclareMathOperator{\letters}{alph}
\DeclareMathOperator{\Pos}{Pos}
\DeclareMathOperator{\loc}{loc}

\def\NP{\textsf{NP}}
\def\Loc{\textsf{Loc}}
\def\Cutwidth{\textsf{Cutwidth}}
\def\Pathwidth{\textsf{Pathwidth}}
\def\nth#1{#1$^{\text{th}}$}

\newif\ifpaper
\paperfalse 

\def\ebs{\textsf{ebs}}

\usetikzlibrary{patterns}
\usetikzlibrary{arrows}
\usetikzlibrary{decorations.pathreplacing}
\usetikzlibrary{arrows.meta}

\begin{document}

\title{Blocksequences of $k$-local Words}
\titlerunning{Blocksequences}

\author{Pamela Fleischmann\inst{1} \and Lukas Haschke\inst{1} \and 
Florin Manea\inst{2} \and Dirk Nowotka\inst{1} \and Cedric Tsatia Tsida\inst{1} 
\and Judith Wiedenbeck\inst{1}}

\authorrunning{P. Fleischmann et al.}

\institute{Kiel University, Germany \and University of Göttingen, Germany\\
 \email{fpa@informatik.uni-kiel.de, stu105615@mail.uni-kiel.de,  florin.manea@informatik.uni-goettingen.de, dn@informatik.uni-kiel.de,  stu111120@mail.uni-kiel.de, stu107029@mail.uni-kiel.de}}

\maketitle
\begin{abstract}
The locality of words is a relatively young structural complexity measure, introduced by Day et al. in 2017 in order to define classes of patterns with variables which can be matched in polynomial time. 
The main tool used to compute the locality of a word is called marking sequence: an ordering of the distinct letters occurring in the respective order. Once a marking sequence is defined, the letters of the word are marked in steps: in the \nth{i} marking step, all occurrences of the \nth{i} letter of the marking sequence are marked. As such, after each marking step, the word can be seen as a sequence of blocks of marked letters separated by blocks of non-marked letters. By keeping track of the evolution of the marked blocks of the word through the marking defined by a marking sequence, one defines the blocksequence of the respective marking sequence. We first show that the words sharing the same blocksequence are only loosely connected, so we consider the stronger notion of extended blocksequence, which stores additional information on the form of each single marked block. In this context, we present a series of combinatorial results for words sharing the extended blocksequence.\end{abstract}

\section{Introduction}
\input{introduction}

\section{Preliminaries and Initial Results}
\input{prelims}

\section{Neighbourless Marking Sequences and a Normal Form}
\input{normalform}

\input{localityofnormalform}

\section{The Case $|\Sigma|\leq 3$}\label{card2}
\input{card3}

\section{Conclusions}
\input{conclusion}

\bibliography{blocksequences.bib}

\ifpaper
\newpage
    \section*{Appendix}
    \input{appendix}
\fi
%
%
%
%

\end{document}

%% file: introduction.tex
The \emph{locality} of words (also called strings) is a structural-complexity measure which has been introduced in~\cite{FSTTCS}. To define the locality of a word several notions are important. Firstly, a {\em marking sequence} for that word is an ordering of the symbols occurring in it. For each {\em marking sequence},  we can mark the letters of the word in steps, as follows: in the \nth{i} marking step, all occurrences of the \nth{i} letter of the marking sequence are marked. As such, after each marking step, the word can be seen as a sequence of blocks of marked letters separated by blocks of non-marked letters. Clearly, after each new marking step of a marking sequence, more symbols become marked, so the marked blocks grow and they may unite. Observing the evolution of the marked blocks leads to the definition of {\em the marking number} of the respective marking sequence: the maximal number of marked blocks which occur in the word after a marking step. {\em The locality number of a word} (for short, {\em locality}) is defined as the minimal marking number over all marking sequences for that word.

More precisely, a word is $k$-local if there exists a marking sequence for the respective word such that after each step of the sequence there are at most $k$ contiguous blocks of marked symbols in the word. The \emph{locality number} (or, for short, locality) of a word is the smallest $k$ for which that word is $k$-local, or, in other words, the minimum marking number over all marking sequences. For instance, if $\tb\ta\tn\ta\tn\ta$ is marked according to the marking sequence $(\tb,\tn,\ta)$ the largest number of marked blocks we get (i.e., the marking number of the sequence) is $3$. Thus, $\tb\ta\tn\ta\tn\ta$ is $3$-local. However, if we take the marking sequence $(\tn,\ta,\tb)$ the largest number of blocks we get is $2$ - and we cannot do better. Thus, $\tb\ta\tn\ta\tn\ta$ has the locality number $2$. The locality number of a word describes how many separated (or isolated) marked regions must at least be maintained in exploring the word w.r.t. possible marking sequences; thus, it can be interpreted as a structural complexity measure (e.g., by associating some cost per marked~region).

The original motivation for the introduction of locality in \cite{FSTTCS} is the fact that patterns with variables which have a low locality can be efficiently matched. A \emph{pattern} is a word that consists of \emph{constant letters}  (e.\,g., $\ta, \tb, \tc$) and \emph{variables} (e.\,g., $x_1, x_2, x_3, \ldots $). A pattern is mapped to a word by uniformly replacing the variables by words with constant letters. For 
example, $x_1 x_1 \ta x_2 x_2$ can be mapped to $\ta \tc \ta \tc \ta \tc \tc$, by replacing $x_1$ by $\ta\tc$ and $x_2$ by $\tc$. If a pattern $\alpha$ can be mapped to the word $w$, we say that $\alpha$ matches $w$. Deciding whether a given pattern matches a given 
word is an important problem with applications in many areas: combi\-natorics on words (word equations~\cite[Chapters 12 and 13]{Loth02}, 
unavoidable patterns~\cite[Chapter 3]{Loth02}), formal-language theory (pattern languages~\cite{ang:fin2}), and learning theory (inductive 
inference~\cite{ang:fin2}, PAC-learning~\cite{kea:apo}), database 
theory (extended conjunctive regular path queries~\cite{bar:exp}), programming languages (the processing of extended regular expressions with backreferences~\cite{Fre2013,fri:mas}, used in programming languages like Perl, Java, Python, etc). In general, the \emph{matching problem} is $\NP$-complete 
\cite{ang:fin2}. This is especially bad for some computational tasks on patterns which 
implicitly solve the matching problem: such problems become, inherently, intractable. One such example is the task of finding descriptive patterns for a set of strings \cite{FeMaMeSc16_TCS}, which is useful in the context of learning theory. 

A thorough analysis of the 
complexity of the matching  problem for patterns of variables was performed~\cite{rei:patIaC,FerSch2015,FerSchVil2015,FeMaMeSc14_stacs} and some classes of patterns admitting polynomial time matching, usually defined by restricting structural parameters, were identified. In \cite{FSTTCS} it was shown that $k$-local patterns can be matched in polynomial time when $k$ is a constant, and, based on the results of \cite{FeMaMeSc16_TCS}, that descriptive $k$-local patterns can be efficiently computed for a given set of strings. 

Thus, the study of the locality of words and patterns seems interesting and well-motivated and the most natural problem one could identify in this area was computing the locality number of a word. The problem $\Loc$ of deciding whether the locality of a given word is upper bounded by a given number $k\in\N$ was shown to be $\NP$-complete in \cite{ICALP2019}. More interestingly, in the same work, strong (and surprising) relations between the string-decision problem $\Loc$ and the graph decision problems $\Cutwidth$ (asking to decide whether the cutwidth of a graph is upper bounded by a given number) and $\Pathwidth$ (asking to decide whether the pathwidth of a graph is upper bounded by a given number) were established. These connections explained, on the one hand, all kinds of algorithmic difficulties arising in solving $\Loc$, and, on the other hand, lead to a state-of-the-art approximation algorithm for computing the cutwidth of graphs. 

\textbf{Our contribution.}
We extend the study of the locality of words by taking a combinatorics-on-words-centric perspective. As explained before, while marking a word with respect to a marking sequence, we obtain after each step a set of factors of the word which consist of marked letters and are bounded by unmarked letters. This set of factors provides a snapshot of the word after each marking step. In our setting the number of marked blocks from each snapshot is important. We will call the sequence of numbers of blocks occurring in these snapshots, in the order in which they occur during the marking sequence, the {\em blocksequence} associated to the marking sequence. Looking again at the word $\tb\ta\tn\ta\tn\ta$ and the marking sequence $(\tb,\tn,\ta)$ we obtain the corresponding blocksequence $(1,3,1)$. 

Now, if we assume that we are only given the sequence $(1,3,1)$, we can trivially tell that this is a blocksequence of a word over a three-letter alphabet. Taking into account that the letters may be renamed we can assume that this three-letter alphabet consists of the letters $\ta$, $\tb$, and $\tc$, and the marking sequence defining the considered blocksequence is $\sigma_{\Sigma}=(\ta,\tb,\tc)$. In other words, we can restrict ourselves to a canonical marking sequence, and our reasoning will be true up to the renaming of all letters. This leads to the question of finding the set of words having the given blocksequence when marked according to $\sigma_{\Sigma}$ and understanding what these words have in common from a combinatorial point of view, e.g., w.r.t. to their locality number. As we have seen the locality of $\tb\ta\tn\ta\tn\ta$ is $2$ and not $3$, so is this a characteristic of all words sharing the blocksequence $(1,3,1)$?
We show that the blocksequence alone does not provide much information, and thus we enrich the blocksequence with more combinatorial information: we do not only store the number of marked blocks in each step, but also the {\em kind} each occurrence of a letter has in the respective step, e.g. neighbour, join, or singleton. In this setting we are able to define a normal form for each class of words having the same extended blocksequence. We show how to obtain the normal form for a given word by defining three rules and we compare the locality of the normal form with the locality of the words from the same class. We finally present, in the case of words over three-letters alphabets, how the optimal marking sequence (the one determining the locality of the word) can be obtained by examining the extended block sequence.

%% file: prelims.tex
\noindent {\bf $\S $ Basic Definitions.} Let $\N$ be the set of natural numbers and $\N_0 = \N\cup\{0\}$.
Let $[n]$ denote the set $\{1,\ldots, n\}$ and $[n]_0 = [n]\cup\{0\}$
for an $n\in\N$.

An alphabet is a finite set $\Sigma=\{\ta_1,\dots,\ta_{\ell}\}$ of $\ell\in\N$ symbols, called {\em letters}.
The alphabet is called {\em ordered} if there exists a total ordering $<$ on the letters. We assume here $\Sigma$ to be ordered with 
$\ta_i<\ta_{i+1}$ for all $i\in[\ell-1]$.
$\Sigma^*$ denotes the set of all finite words over $\Sigma$, i.e. the free monoid over $\Sigma$. 
The {\em empty word} is denoted by $\varepsilon$ and 
$\Sigma^+=\Sigma^*\backslash\{\varepsilon\}$.
The length of a word $w$ is denoted by $|w|$. Define 
$\Sigma^{k}:=\{w\in\Sigma^{\ast}|\,|w|= k\}$ for a $k\in\N$. The number of occurrences of a 
letter $\ta\in\Sigma$ in a word $w\in\Sigma^{\ast}$ is denoted by $|w|_\ta$. Define the set of letters occurring in $w\in\Sigma^{\ast}$ by $\letters(w)=\{\ta\in\Sigma|\,|w|_{\ta}>0\}$.
The \nth{i} letter of a word $w$ is given by $w[i]$ for $i\in[|w|]$. For a given 
word $w\in\Sigma^n$ the {\em reversal} of $w$ is defined by 
$w^R=w[n]w[n-1]\cdots w[2]w[1]$. The powers of $w\in\Sigma^{\ast}$ are 
defined recursively by $w^0=\varepsilon$, $w^n=ww^{n-1}$ for $n\in\N$.
A word $u\in\Sigma^{\ast}$ is a \emph{factor} of $w\in\Sigma^{\ast}$, if 
$w=xuy$ holds for some words $x,y\in\Sigma^{\ast}$. Moreover, $u$ is a 
\emph{prefix} (resp., \emph{suffix}) of $w$ if $x=\varepsilon$ (resp., $y=\varepsilon$) holds. The factor $w[i]w[i+1]\cdots w[j]$ of $w$ is denoted by $w[i..j]$, for $1\leq i\leq j\leq |w|$. 
Given a property $P:\Sigma\rightarrow\{0,1\}$, a factor $u$ is a {\em $P$-block}  of a word $w=xuy$ if $P(u[i])=1$ for all 
$i\in[|u|]$ and $P(x[|x|])=P(y[1])=0$ (if $x$ or $y$ are empty the constraint does not have to be fulfilled). 
For the property $P_\ta$ defined by $P_{\ta}(x)=1$ iff $x=\ta$ for $x\in\Sigma$, the
word $\ta\tb\ta\ta\ta\tb\ta\ta\tb\tb$ has 3 $P_{\ta}$-blocks (or short three $\ta$-blocks).

In the following, we give the main definitions on $k$-locality, following~\cite{FSTTCS}. 
 
\begin{definition}\label{basedefs}
Let $\overline{\Sigma}=\{\overline{x}\mid x\in \Sigma\}$ be the set of \emph{marked letters}. For a word $w\in\Sigma^{\ast}$, a {\em marking sequence} of the 
letters occurring in $w$, is an enumeration $(x_1,x_2,\ldots, x_{|\letters(w)|})$ of $\letters(w)$. We say that $\ta_i\leq_{\sigma}\ta_j$ if $\ta_i$ occurs before $\ta_j$ in $\sigma$, for $i,j\in[|\letters(w)|]$. The enumeration obeying the total order of the alphabet is called the {\em canonical marking sequence} $\sigma_{\Sigma}$. A letter $x_i$ is called {\em marked at stage $k\in\N$}  if $i\leq k$.
 Moreover, we define 
$w_k$, {\em the marked version of $w$ at stage $k$}, as the word obtained from $w$ 
by replacing all $x_i$ with $i\leq k$ by $\overline{x_i}$. A factor of $w_k$ is a \emph{marked block} if the defining property of the block is that it contains only elements from $\overline{\Sigma}$. 
The locality of a word $w$ w.r.t. a marking 
sequence $\sigma$ ($\loc_{\sigma}(w)$) is the maximal number of marked blocks 
that occurred during the marking process.
\end{definition}

In the context of Definition \ref{basedefs}, $w_{|\letters(w)|}$ is always completely marked.
Using the idea of a marking sequence, we define the $k$-locality of a word.
 
\begin{definition}\label{def1}
A word $w\in\Sigma^{\ast}$ is {\em $k$-local} for $k\in\N_0$ if 
there exists a marking sequence 
$(x_1,\ldots,x_{|\letters(w)|})$ of $\letters(w)$, such that, for all $i\leq |\letters(w)|$ we have that 
$w_i$ at stage $i$, has at most $k$ marked blocks. A word is called {\em strictly} 
$k$-local if it is $k$-local but not $(k-1)$-local. 
\end{definition}

Consider the word $\tb\ta\tn\ta\tn\ta\in\{\ta,\tb,\tn\}^{\ast}$. The marking sequence $(\ta,\tb,\tn)$
leads to the sequence $w_1=\tb\overline{\ta}\tn\overline{\ta}\tn\overline{\ta}$ (3 marked blocks), $w_2=\overline{\tb\ta}\tn\overline{\ta}\tn\overline{\ta}$ (3 marked blocks), and $w_3=\overline{\tb\ta\tn\ta\tn\ta}$ (1 marked block), i.e. $\tb\ta\tn\ta\tn\ta$ is $3$-local. In fact, it is strictly $2$-local witnessed by the marking sequence $(\tn,\ta,\tb)$ (it is not $1$-local, since this would imply to start with marking $\tb$ and either marking afterwards $\ta$ or $\tn$ leads to more than one marked block). As a second example consider the word $\ta^3\tb^4$. This word is $1$-local since for both marking sequences $(\ta,\tb)$ and $(\tb,\ta)$ the blocks of letters are marked in one step. This motivates to consider the notion of the print of a word - or condensed word - introduced in \cite{DBLP:journals/fuin/SerbanutaS06} and \cite{ICALP2019}, respectively.

\begin{definition}
For $w=x_1^{k_1}x_2^{k_2}\dots x_{m}^{k_{m}}\in\Sigma^{\ast}$ with $k_i,m\in\N$, $i\in[m]$, and $x_j\neq x_{j+1}$ for $j\in[m-1]$, the {\em print (condensed form)} of $w$ is defined by $x_1\dots x_{m}$. A word is called {\em condensed}, if it is its own print.
\end{definition}

\medskip

\noindent {\bf $\S$ Initial Results.} Since in our setting the multiplicity of single letters does not affect the results (all these letters form a single marked block), we restrict the setting to condensed words implicitly, i.e. each $w\in\Sigma^{\ast}$ is implicitly meant to be condensed. Now we define the notion of the blocksequence that captures the number of marked blocks during the marking process. Moreover,  we assume $\letters(w)=\Sigma$.

\begin{definition}
Let $w\in\Sigma^{\ast}$ and $\sigma=(y_1,\dots,y_{\ell})$ be a marking sequence. The {\em blocksequence} $\beta_{\sigma}(w)$
is the sequence $(b_1,\dots,b_{\ell})$ over $\N$ such that in $\sigma$'s \nth{$i$} 
stage on marking $w$, $b_i$ blocks are marked, for all $i\in[\ell]$. 
\end{definition}

Coming back to $\tb\ta\tn\ta\tn\ta$, the marking sequence $(\ta,\tb,\tn)$ leads to the blocksequence 
$(3,3,1)$ and the marking sequence $(\tn,\ta,\tb)$ to $(2,1,1)$. 
Since $|w|_{|\letters(w)|}$ is one marked block, the last position in a blocksequence has to be $1$ and moreover the first position is exactly $|w|_{y_1}$.
Changing the perspective, $n$-tuples (with the last position being $1$) can be seen as 
a blocksequence w.r.t. the canonical marking sequence given by the alphabet and its order. This point of view is inspired by the idea to group words with the same blocksequence in order to deduce information about their locality.

\begin{definition}
For a given $\ell$-tuple $\beta=(b_1,\dots,b_{\ell-1},1)$ define the set of words that 
give exactly $\beta$ on marking with $\sigma_{\Sigma}$ by $\mathfrak{W}_{\beta}=\{w\in\Sigma^{\ast}|\,
\beta_{\sigma_{\Sigma}}(w)=\beta\}$. 
\end{definition}

First, we prove that the class $\mathfrak{W}_{\beta}$ is not empty for all $\beta=(b_1,\dots,b_{\ell-1},1)$.

\ifpaper 
\begin{theorem}[$\ast$]
\else    
\begin{theorem}
\fi
\label{ntupel}
For all $\beta=(b_1,\dots,b_{\ell-1},1)\in\N^{\ell}$ there exists $n_\beta\in\N$ such that for all $n\geq n_\beta$ we have $\Sigma^{n}\cap \mathfrak{W}_{\beta}\neq \emptyset$ and for all $m< n_\beta$ we have $\Sigma^{m}\cap \mathfrak{W}_{\beta}=\emptyset$.
\end{theorem}
\ifpaper 
\else    
    \input{proof_ntupel}
\fi

In fact, one can characterise precisely the set $\Sigma^{n_\beta}\cap \mathfrak{W}_{\beta}$, as well as the condensed words from $ \mathfrak{W}_{\beta}$\ifpaper (see Theorem \ref{thm:condensed} in Appendix). 
\else .
\input{proof_condensed}
\fi

A blocksequence induced by $\sigma_{\Sigma}$ does not determine a word uniquely witnessed by $\ta\tb\tc\tb\ta$ and $\ta\tb\tc\ta$ for the blocksequence $(2,2,1)$. In fact, it does not even determine a print of the words sharing the same blocksequence uniquely. Indeed, for $\beta=(3,6,1)$ the words 
$w=\ta\tc\tb\tc\tb\tc\ta\tc\tb\tc\ta$, $w'=\ta\tb\ta\tb\ta\tc\tb\tc\tb\tc\tb\tc\tb\tc\tb$, and  
$w''=\ta\tc\tb\tc\ta\tc\tb\tc\ta\tc\tb$
are in $\mathfrak{W}_{\beta}$. Moreover, when considering a different marking sequence, these words have different blocksequences. For instance, $(\tc,\ta,\tb)$  leads to the blocksequences
$(5,4,1), (5,7,1)$, and $(5,3,1)$, respectively. Thus, it is to be expected that even if some words have the same blocksequence w.r.t. a marking sequence, they may have different blocksequences w.r.t. other marking sequences, and, consequently, different localities. 
The main difference in the above words are the different roles the letters have: in $w$ the occurrences of $\tb$ are between two occurrences 
of $\ta$ but $\tb$ does not join the $\ta$-blocks whereas in $w'$ all {\em gaps} between 
the $\ta$s are closed by join occurrences of $\tb$; in $w''$ in each {\em gap} between $\ta$s is only one 
occurrence of $\tb$. This observation leads to the following differentiation of
occurrences of letters:  when the letter is marked it may occur adjacent to exactly 
one marked block (neighbouring), it may join two 
blocks (joining), or it may not be adjacent to any marked block (singleton).
Notice that different occurrences of letters may have different roles.

\begin{definition}\label{rules}
Let $\sigma=(y_1,\dots,y_{\ell})$ be a marking sequence of $w\in\Sigma^{\ast}$. 
At stage $i\in[\ell]$, an occurrence of $y_i$ is said to be a \\
- {\em neighbour} if there exist $u_1\in\overline{\Sigma}^{+}$, 
$u_2\in\Sigma^{+}$ and $v_1,v_2\in(\Sigma\cup\overline{\Sigma})^{\ast}$ with 
$w_i=v_1u_1y_iu_2v_2$, $w_i=v_1u_2y_iu_1v_2$, $w_i=v_1u_1y_i$, or $w_i=y_iu_1v_1$,\\
- {\em join} if there exist $u_1,u_2\in\overline{\Sigma}^{+}$
and $v_1,v_2\in(\Sigma\cup\overline{\Sigma})^{\ast}$ with 
$w_i=v_1u_1y_iu_2v_2$,\\
- {\em singleton} if there exist $u_1,u_2\in\Sigma^{+}$, 
and $v_1,v_2\in(\Sigma\cup\overline{\Sigma})^{\ast}$ with 
$w_i=v_1u_1y_iu_2v_2$, $w_i=v_1u_1y_i$, or $w_i=y_iu_2v_2$.\\
A marking sequence $\sigma$ is called {\em neighbourless} for a word $w\in\Sigma^{\ast}$ if in any stage while marking $w$ with $\sigma$ no neighbour occurrences exist. A word $w\in\Sigma^{\ast}$ is called {\em neighbourless} if there exists a neighbourless marking sequence $\sigma$ for $w$.
\end{definition}

Another observation of $w'$ and $w''$ leads to different forms of singletons: the ones occurring between previously marked letters and the ones occurring outside.

\begin{definition}
Let $\sigma=(y_1,\dots,y_{\ell})$ be a marking sequence of $w\in\Sigma^{\ast}$.
The {\em core} of $w$ at stage $i\in[\ell]_{>1}$ is defined as $u\in\Fact(w)$ with $w_i=v_1uv_2$,
$\letters(v_1),\letters(v_2)$ $\subseteq\{y_{i},\dots,y_{\ell}\}$, and $u[1],u[|u|]\in\{y_1,\dots,y_{i-1}\}$.
A singleton occurrence at stage $i\in[\ell]$ of a letter $y_i\in\Sigma$  is called 
{\em separating} (or a {\em separator}) if it is of the form $v_1u_1z_1y_i z_2 u_2v_2$ 
with $u_1,u_2\in\overline{\Sigma}^+$, 
$v_1,v_2\in(\Sigma\cup\overline{\Sigma})^{\ast}$, and $z_1,z_2\in\Sigma^+$.
A singleton occurrence that is not a separator is called {\em satellite}.
\end{definition}

\begin{remark}\label{arrowNotation}
Separators are within the core whereas satellites are to the left or to the right of the core. For convenience we introduce for a given marking sequence $(y_1,\dots,y_{\ell})$ of a word $w\in\Sigma^{\ast}$ the notations $\overrightarrow{y_i}$ and $\overleftarrow{y_i}$ as arbitrary elements from the sets 
$\{y_ix|\,x\in\{y_{i+1},\dots,y_{\ell}\}^+\}$ and $\{xy_i|\,x\in\{y_{i+1},\dots,y_{\ell}\}^+\}$ respectively, for all $i\in[\ell]$, if $\overrightarrow{y_i}$ ($\overleftarrow{y_i}$ resp.) is a factor $w[j_1\dots j_2]$ of $w$ 
and additionally with $w[j_2+1]\leq_{\sigma}y_i$ ($w[j_1-1]\leq_{\sigma}y_i$ resp.) if $\overrightarrow{y_i}$ ($\overleftarrow{y_i}$ resp.) is not a suffix (prefix resp.) of $w$. By $\overrightarrow{y_i}^m$ we denote a word containing of $m\in\N_0$ possibly different occurrences of $\overrightarrow{y_i}$ (analogously for $\overleftarrow{y_i}$), e.g. $\tb\tc\td\tb\tc$ could be abbreviated by $\overrightarrow{\tb}^2$.
With this notation, the palindromic structure of $k$-local words already mentioned in \cite{FSTTCS} becomes clearer
in this context: if $c_i$ is the core at stage $i$, the word is of the form 
$\overrightarrow{y_\ell}^{s_\ell}\dots \overrightarrow{y_i}^{s_i} c_i \overleftarrow{y_i}^{r_1}\dots
\overleftarrow{y_\ell}^{r_{\ell}}$
with $s_j,r_j\in\N_0$, $j\in[l]$.
\end{remark}

As we have seen, the blocksequence does not provide much information. Therefore we refine the sets $\mathfrak{W}_{\beta}$ by sequences containing, as well, information on the different types of occurrences we have just introduced. 

\begin{definition}
Let $\sigma=(y_1,\dots,y_{\ell})$ be a marking sequence of $w\in\Sigma^{\ast}$. Define the {\em join-sequence} $\iota_w=(j_1,\dots,j_{\ell-2})$ such that $j_i$ is the number of join occurrences of $y_{i+1}$ and the {\em separator-sequence} $\zeta_w=(s_1,\dots,s_{\ell-2})$ such that $s_i$ is the number of separating occurrences of $y_{i+1}$. Finally define the {\em extended blocksequence} ($\ebs$) by $\gamma_w=(\beta_w,\iota_w,\zeta_w)$ w.r.t. $\sigma$.
\end{definition}

Consider the word $w=\ta\tb\ta\td\tb\tc\tb\td\ta\tc\tb\td\tc$ marked with $\sigma_{\Sigma}$. Regarding $\tb$, $w[2]$ joins the two $\ta$ in stage 2 and $w[5],w[7]$ are separating the $\ta$s at positions $3$ and $9$. The $\tc$s at position $6$ and $10$ join two marked blocks in stage $3$ but no occurrence of $\tc$ separates two marked blocks. This leads to $\beta=(3,5,4,1)$. Moreover we have $\iota=(1,2)$ and $\zeta=(2,0)$ as join and separating sequence, resp.

\begin{remark}
For a blocksequence $\beta=(b_1,\dots,b_n)$ it suffices to state $n-2$ elements of $\iota$ and $\zeta$ explicitly. Since we only consider condensed words the first letter to be marked creates $b_1$ separate blocks. These occurrences are all satellites. Similarly the last letter joins all remaining gaps between the marked blocks.
\end{remark}

Whereas by Theorem~\ref{ntupel} for each sequence of natural numbers ending with $1$ there exists a word having this sequence as a blocksequence, the same does not hold for $\ebs$: Consider $((2,1,1),(0),(5))$ as an $\ebs$ for a ternary word. Thus we have two occurrences of $\ta_1$. Moreover, we know that marking $\ta_1$ and $\ta_2$ leads to one block and consequently the two $\ta_1$ need to be joined but the join-sequence dictates that we do not have a join-occurrence of $\ta_2$. So there is no word with this $\ebs$, which leads to the introduction of the notion of valid $\ebs$.

\begin{definition}
A triple $\gamma$ of sequences over natural numbers is called a {\em valid $\ebs$}, if there exists a word $w\in\Sigma^{\ast}$ with $\gamma=\gamma_w$.
\end{definition}

Very importantly, we can exactly identify the valid $\ebs$.
\ifpaper 
\begin{theorem}[$\ast$]
\else    
\begin{theorem}
\fi
\label{validity}
A triple $\gamma=(\beta,\iota,\zeta)$ of sequences $\beta=(b_1,\dots,b_{\ell})$, $\iota=(j_1,\dots$, $j_{\ell-2})$, and $\zeta=(s_1,\dots, s_{\ell-2})$ for $\ell\in\N_{\geq 2}$ is a valid $\ebs$ w.r.t. $\sigma_{\Sigma}$ iff $b_{\ell}=1$, $\max\{b_i-b_{i+1},0\}\leq j_{i}\leq b_{i}-1$, and $s_i=0$ if $b_{i+1}-b_{i}-j_i=0$ as well as $b_{i}-j_i+s_i\leq b_{i+1}$ for all $i\in[\ell-2]$. 
\end{theorem}
\ifpaper
\else
    \input{proof_validity}
\fi

\begin{remark}
If $b_i=b_{i+1}$ holds for an $\ebs$ $\gamma$,  the number of occurrences joining existing blocks and singletons creating new blocks of the letter $\ta_{i+1}$ has to be equal and $s_i\leq j_i$. If $j_i = b_i - 1$ (all blocks are joined) then $s_i=0$ and there has to be exactly $j_i$ satellites. For the special case that $b_i=b_{i+1}=1$ there is only one block before marking $\ta_{i+1}$ and there cannot be any joins or separators ($j_i=s_i=0$). Since $b_{i+1}=1$ there can be no satellite occurrence as well and therefore $\ta_{i+1}$ can only occur as a neighbour.
\end{remark}

\begin{definition}
For a valid $\ebs$ $\gamma$ set $\mathfrak{V}_{\gamma}=\{w\in\Sigma^{\ast}|\,\gamma_w=\gamma\}$ and define the equivalence relation $u\sim_{\gamma} v$ if $\gamma_u=\gamma_v$ w.r.t. a given marking sequence $\sigma$.
\end{definition}

Moreover, for a valid $\ebs$ $\gamma$, we can show that all words in $\mathfrak{V}_{\gamma}$ have the same length (in contrast to the words of $\mathfrak{W}_{\beta}$, cf. Theorem~\ref{ntupel}). This will allow us to define later a normal form for each valid $\ebs$.

\ifpaper
\begin{theorem}[$\ast$]
\else
\begin{theorem}
\fi
\label{lengthforebs}
For a valid $\ebs$ $\gamma=((b_1,\dots,b_{\ell-1},1),(j_1,\dots,j_{\ell-2}),\zeta)$, all words in $\mathfrak{V}_{\gamma}$ have length $b_1+b_{\ell-1}-1+\sum_{i=1}^{\ell-1}(b_i-b_{i-1}+2j_{i-1})$.
\end{theorem}

\ifpaper
\else
\input{proof_lengthforebs}
\fi

We finish this section with a result about neighbourless marking sequences which will be of importance in the following sections.

\ifpaper
\begin{lemma}[$\ast$]
\else
\begin{lemma}
\fi
\label{occletter}
For a given $\ebs$ $\gamma$, if all occurrences of the letters are 
either join-occurrences or singletons, $|w| = |w'|$ and $|w|_x=|w'|_x$ holds for all neighbourless
$w,w'\in\mathfrak{V}_{\gamma}$ and all $x\in\Sigma$.
\end{lemma}
\ifpaper
\else
    \input{proof_occletter}
\fi

%% file: proof_ntupel.tex
\begin{proof}
Set $n_\beta = \sum_{i\in [\ell]} \Delta(b_i,b_{i-1})$ with $b_0=0, b_{\ell}=1$ and the notation\[
\Delta(x,y)=
\left\{
	\begin{array}{ll}
		y-x  & \mbox{if } y>x, \\
		1 & \mbox{if } x= y, \\
		x-y  & \mbox{if } x>y
	\end{array}
\right.
\]
for $x,y\in\N$. We will show that there exists a word $w\in\Sigma^{n_{\beta}}$  whose blocksequence w.r.t. the canonical marking sequence $\sigma_{\Sigma}$ is $\beta$. This word is defined by the following algorithm. Define in the first step the blocks $u_{1,1}=\ta_1, \ldots ,u_{1,b_1}=\ta_1$. For each $i\in\{2,\ldots,\ell\}$ do the following: 
\begin{itemize}
\item if $b_i=b_{i-1}$, let $u_{i,1}=u_{i-1,1}\ta_i$, and let $u_{i,j}=u_{i-1,j}$ for $j\in \{2,\ldots, b_{i}\}$,
\item if $b_i<b_{i-1}$ and $d=\Delta(b_i,b_{i-1})=b_{i-1}-b_i$, let 
$$u_{i,1}=u_{i-1,1}\ta_i\cdots u_{i-1,d} a_i u_{i-1,d+1}$$
and $u_{i,j}=u_{i-1,j+d}\mbox{ for }j\in \{2,\ldots, b_{i-1}-d\}$,
\item if $b_i>b_{i-1}$ and $d=\Delta(b_i,b_{i-1})=b_{i}-b_{i-1}$, let 
$$u_{i,j}=u_{i-1,j},\mbox{ for }j\in \{1,\ldots, b_{i-1}\}\mbox{ and }u_{i,b_{i-1}+j}=\ta_{i}\mbox{ for }j\in [d].$$
\end{itemize}
Finally, let $w=u_{\ell,1}$. It is immediate that $|w|=n_\beta$ and its blocksequence w.r.t. $\sigma_{\Sigma}$ is $\beta$.
Notice that $w$ is not necessarily the unique word in  $\Sigma^{n_\beta}\cap \mathfrak{W}_{\beta}.$
For all $w'\in\Sigma^{\ast}$ with condensed form $w$, the blocksequence of $w'$ w.r.t. $\sigma_{\Sigma}$ is $\beta$. Thus, for all $n\geq n_\beta$ we have $\Sigma^{n}\cap \mathfrak{W}_{\beta}\neq \emptyset$.

It remains to show that for all $m< n_\beta$ we have $\Sigma^{m}\cap \mathfrak{W}_{\beta}=\emptyset$.
Let $u$ be a word whose blocksequence w.r.t. the marking sequence $\sigma_{\Sigma}$ is $\beta$, i.e. $u\in  \mathfrak{W}_{\beta}$. We will show by induction on $i$ that $\sum_{j\in [i]}|u|_{\ta_j} \geq  \sum_{j\in [i]} \Delta(b_j,b_{j-1})$.  
The property holds clearly for $i=1$. Assume that it holds for all $i-1$. We will show it for $i$ by case analysis. \\
\textbf{case 1:} $b_i=b_{i-1}$\\
Since we have $|u|_{\ta_i}\geq 1$ and $\sum_{j\in [i-1]}|u|_{\ta_j} \geq \sum_{j\in [i-1]} \Delta(b_j,b_{j-1})$, we immediately get $\sum_{j\in [i]}|u|_{\ta_j} \geq \sum_{j\in [i]} \Delta(b_j,b_{j-1})$ by $\Delta(b_i,b_{i-1})=1$. \\
\textbf{case 2:} $b_i>b_{i-1}$\\
One needs to have $|u|_{\ta_i}\geq b_i - b_{i-1}$, as otherwise we could not produce $b_i - b_{i-1}$ new blocks by marking the letters $\ta_i$.  As $\sum_{j\in [i-1]}|u|_{\ta_j} \geq \sum_{j\in [i-1]} \Delta(b_j,b_{j-1})$, we immediately have $\sum_{j\in [i]}|u|_{\ta_j} \geq \sum_{j\in [i]} \Delta(b_j,b_{j-1})$, because $\Delta(b_i,b_{i-1})=b_i - b_{i-1}$. \\
\textbf{case 2:} $b_i<b_{i-1}$\\ 
Here, one needs to decrease the number of blocks, and this means that some blocks need to be joined. More precisely, one needs to decrease the number of blocks to $b_i$ so at least $b_{i-1}-b_{i}+1$ blocks of the existing marked blocks need to be joined. For this, we need $b_{i-1} - b_{i}$ letters $\ta_i$, so $|u|_{\ta_i}\geq b_{i-1} - b_{i}$. As $\sum_{j\in [i-1]}|u|_{\ta_j} \geq \sum_{j\in [i-1]} \Delta(b_j,b_{j-1})$, we immediately have $\sum_{j\in [i]}|u|_{\ta_j} \geq \sum_{j\in [i]} \Delta(b_j,b_{j-1})$, because $\Delta(b_i,b_{i-1})=b_{i-1} - b_{i}$. \\
This concludes our proof, as we get that $|u|\geq n_\beta$.  \qed
\end{proof}

%% file: proof_condensed.tex
\begin{theorem}\label{thm:condensed}
For a blocksequence $\beta=(b_1,\dots,b_{\ell-1},1)\in\N^{\ell}$, we can define exact procedures that enumerate
\begin{itemize}
\item[A.] the words in $\Sigma^{n_\beta}\cap \mathfrak{W}_\beta$,
\item[B.] all condensed words in $\mathfrak{W}_\beta$.
\end{itemize}
\end{theorem}
\begin{proof}
{\bf We first show item A. }

The idea of this proof follows closely the proof of Theorem \ref{ntupel}, so we will use the same notations as in the respective proof. 
In the respective proof, we have constructed a word $w\in\Sigma^{n_{\beta}}$ whose blocksequence w.r.t. the canonical marking sequence $\sigma_{\Sigma}$ is $\beta$. 

We can extend the respective construction to construct any word in $\Sigma^{n_\beta}\cap \mathfrak{W}_\beta$. For simplicity, we describe this as a non-deterministic algorithm.

Just like before, define in the first step the blocks $u_{1,1}=\ta_1, \ldots ,u_{1,b_1}=\ta_1$. Then, we generate nondeterministically a word from the set  $\Sigma^{n_\beta}\cap \mathfrak{W}_\beta$ as follows. For $i\in\{2,\ldots,\ell\}$, in increasing order, do: 
\begin{itemize}
\item if $b_i=b_{i-1}$, choose a $j\in [b_1]$; let $u_{i,j}=u_{i-1,j}\ta_i$, and let $u_{i,t}=u_{i-1,t}$ for $t\in [b_i]\setminus \{t\}$,
\item if $b_i<b_{i-1}$ and $d=\Delta(b_i,b_{i-1})=b_{i-1}-b_i$, choose nondeterministically $d$ pairs of blocks $(u_{i-1,t_j}, u_{i-1,t_j+1})$ for $j\in [d]$ (and assume that they are ordered w.r.t. their second index).

We now process the list of blocks $u_{i-1,g}$, with $g\in [b_{i-1}]$ as follows. For $j$ from $1$ to $d$, concatenate the block ending with $u_{i-1,t_j}$, a letter $\ta_i$, and the block starting with $u_{i-1,t_j+1})$ (these blocks were consecutive in our list).

The list of blocks we obtain this way has $b_i$ elements. We define $u_{i,g}$ as the \nth{$g$} block of this list.

\item if $b_i>b_{i-1}$ and $d=\Delta(b_i,b_{i-1})=b_{i}-b_{i-1}$, let first 
$$u'_{i,j}=u_{i-1,j},\mbox{ for }j\in \{1,\ldots, b_{i-1}\}\mbox{ and }u'_{i,b_{i-1}+j}=\ta_{i}\mbox{ for }j\in [d].$$

Define two ordered lists: $L_1$ is the list of the words $u'_{i,j}\mbox{ for }j\in \{1,\ldots, b_{i-1}\}$ (ordered left to right increasingly w.r.t. the index $j$); $L_2$ is the list of the words $u'_{i,b_{i-1}+j}\mbox{ for }j\in [d]$ (ordered left to right increasingly w.r.t. the index $j$). 

For $j\in [b_i]$, choose one of the lists $L_1$ and $L_2$ nondeterministically, remove its first element $u$, define $u_{i,j}=u$.  
\end{itemize}

Finally, let $w=u_{\ell,1}$. It is immediate that $|w|=n_\beta$ and its blocksequence w.r.t. $\sigma_{\Sigma}$ is $\beta$.

Assume now that there is $w\in \Sigma^{n_\beta}\cap \mathfrak{W}_\beta$ that cannot be obtained by this procedure. Clearly, $w$ is a condensed word (otherwise, a word shorter than $n_\beta$ would be in $\mathfrak{W}_\beta$, contradiction). Following the proof of Theorem \ref{ntupel}, we get that in the \nth{$i$} step of the marking sequence $\beta$ on $w$ we need to mark exactly $\Delta(b_i,b_{i-1})$ letters $\ta_i$, for $i\geq 2$. So, we execute the marking sequence $\beta$ on $w$. In the first step, we mark exactly $b_1$ letter $\ta_i$, and no two of them occur on consecutive positions of $w$. Now, assume that till step $i-1$, the blocks we marked can be obtained by the nondeterministic process we described. We now move to step $i$. Assume that $w$ has the marked blocks $u_{i-1,j}$, for $j\in [b_{i-1}]$. If $b_i=b_{i-1}$, we have that the number of letters $\ta_i$ contained in $w$ is exactly $1$. The only possibility is that this letter $\ta_i$ occurs next to an already existing block. If $b_i<b_{i-1}$, we have that the number of letters $\ta_i$ contained in $w$ is exactly $b_{i-1}-b_i$. The only possibility is that all these letter $\ta_i$ connect already existing blocks. If $b_i>b_{i-1}$, we have that the number of letters $\ta_i$ contained in $w$ is exactly $b_{i}-b_{i-1}$. The only possibility is that all these letter $\ta_i$ create new blocks, not connected to the existing ones. But in all cases, the blocks can be created by our nondeterministic algorithm. By induction, we get that $w$ is generated by our algorithm.

To enumerate all the words in $\Sigma^{n_\beta}\cap \mathfrak{W}_\beta$, it is enough to implement our nondeterministic algorithm using backtracking.

{\bf This concludes the proof of item A.}

\medskip 

{\bf We can now show item B.} The proof is quite similar. The only major difference is that, when generating the condensed words in $\mathfrak{W}_{\beta}$ by a nondeterministic algorithm as above, we do not know the exact number of letters we need to insert in each step, but only that they need to be more at least $\Delta(b_i,b_{i-1})$ (in step $i$), and they should not create powers (e.g., $\ta_i^2$ in step $i$).

The algorithm is the following. We define in the first step the blocks $u_{1,1}=\ta_1, \ldots ,u_{1,b_1}=\ta_1$; in this case we do not have any other choice. Then, we generate nondeterministically a condensed word from the set $\mathfrak{W}_\beta$ as follows. For $i\in\{2,\ldots,\ell\}$, in increasing order, do: 
\begin{itemize}
\item Start with the blocks $u_{i-1,j}$, with $j\in [b_{i-1}]$.
\item We first nondeterministically choose several groups of two or more consecutive blocks (of the existing blocks) and concatenate them (in the same order) also putting $\ta_i$ letters between each two blocks; then, we choose several of the current blocks, and we concatenate single $\ta_i$-letters at both their ends; finally, we create several blocks consisting of a single letter $\ta_i$ each, so that in the end we have exactly $b_i$ blocks.
\item The blocks obtained in this way are now the blocks $u_{i,j}$.
\end{itemize}
Clearly, one needs to be careful in making sure that when we want to decrease the number of blocks when moving from step $i-1$ to step $i$, we first concatenate enough blocks so that we have at most $b_i$ blocks after the first nondeterministic step.

Finally, let $w=u_{\ell,1}$.

It is clear that a word obtained by our procedure is in ${\mathfrak W}$ and is condensed. By a proof similar to that from item A, we can show by induction that any condensed word in $\mathfrak{W}_{\beta}$ can be obtained by our nondeterministic algorithm.

To enumerate all the words in $\mathfrak{W}_\beta$, it is enough to implement our nondeterministic algorithm using backtracking.\qed
\end{proof}

%% file: proof_validity.tex
\begin{proof}
Consider firstly $\gamma$ to be a valid extended blocksequence. Then there exists $w\in\Sigma^n$ with 
$\gamma=\gamma_{\sigma_{\Sigma}}(w)$. Marking $w$ with $\sigma_{\Sigma}$ leads to a single marked block in the end and thus we have $b_{\ell}=1$. Let $i\in[\ell-2]$. At stage $i$ we have $b_{i}$ marked blocks and thus $b_{i}-1$ gaps between the marked blocks. This implies that $w$ marked with $\sigma_{\Sigma}$ cannot have more than $b_{i}-1$ join occurrences of $\ta_{i+1}$, i.e. $j_{i}\leq b_{i}-1$. If $b_{i+1}<b_i$, the number of blocks is decreased by marking $\ta_{i+1}$, i.e. blocks in stage $i$ need to be joined. This implies $j_i\geq b_i-b_{i+1}$. If $b_{i+1}\geq b_i$, the number of blocks is increased or remains the same by marking $\ta_{i+1}$, i.e. $w$ does not need to have join-occurrences of $\ta_{i+1}$.  By $j_i\in\N_0$ we get $\max\{b_i-b_{i+1},0\}\leq j_i$. Marking $w$ with $\sigma_{\Sigma}$ leads to $b_i$ marked blocks in stage $i$ and $b_{i+1}$ marked blocks in stage $i+1$. Thus if $b_{i+1}-b_i-j_i=0$, after marking only the join-occurrences of $\ta_{i+1}$ leads to $b_{i+1}$ blocks. By the definition of separating occurrences follows that marking such an occurrence increases the number of marked blocks. This implies that $w$ does not have separating occurrences of $\ta_{i+1}$, namely $s_i=0$. By the same argument we get that the blocks marked at stage $i+1$ is lower bounded by the amount of blocks marked at stage $i$ minus the join occurrences of $\ta_{i+1}$ plus the separating occurrences of $\ta_{i+1}$. This concludes the first direction.

\medskip

Consider now $\gamma=(\beta,\iota,\zeta)$ with $\beta=(b_1,\dots,b_{\ell})$, $\iota=(j_1,\dots,j_{\ell-2})$, and $\zeta=(s_1,\dots,s_{\ell-2})$ for $\ell\in\N_{\geq 2}$ and the four constraints. We prove that $\gamma$ is a valid extended blocksequence by constructing $w\in\Sigma^{\ast}$ with $\gamma=\gamma_{\sigma_{\Sigma}}(w)$ inductively. Let $\bullet$ be a symbol different from all letters. Define $w_1=(\ta_1\bullet)^{b_1}$.
Let $i\in[\ell-2]_{>1}$ and assume $w_i$ to be constructed. Define $w_{i+1}$ in the following way: firstly replace the first $j_i$ occurrences of $\bullet$ by $\ta_{i+1}$. Then replace the next $\bullet$ by $(\bullet \ta_{i+1})^{s_i}\bullet$. Denote the obtained word by $v$. Now define $w_{i+1}=v$ if $b_{i+1}-(b_i-j_i+s_i)=0$ or as $v(\bullet \ta_{i+1})^{b_{i+1}-(b_i-j_i+s_i)}$ otherwise. For obtaining $w$ from $w_{\ell-1}$ replace all occurrences of $\bullet$ by $\ta_{\ell}$. Notice that $w$ is well-defined by the four constraints. If $w$ is marked with $\sigma_{\Sigma}$ in the first stage we mark $b_1$ blocks since $\ta_1$ only occurs in $w_1$ and is never added in a later step of the construction. Assume that in stage $i$ we have $b_i$ marked blocks, $j_i$ join occurrences and $s_i$ separating occurrences of $\ta_i$. By the definition of $w_{i+1}$ we replaced the first $j_i$ occurrences of $\bullet$ by $\ta_{i+1}$. Since in $w_i$ everything but the occurrences of $\bullet$ are marked in stage $i$, in $w$ we have a single marked block as a prefix that includes all these join-occurrences of $\ta_{i+1}$ including the neighbouring block to the right. The next $\bullet$ in $w_i$ was replaces by $s_i$ separating occurrences of $\ta_{i+1}$ and thus we mark another $s_i$ blocks. If $b_{i+1}-(b_i-j_i+s_i)=0$ we have exactly $b_{i+1}$ marked blocks in $w$ since we did not add anything to $w_{i}$. If $b_{i+1}-(b_i-j_i+s_i)>0$, we added $b_{i+1}-(b_i-j_i+s_i)$ occurrences of $\ta_{i+1}$ to $v$ and thus we have in stage $i+1$ in $w$ exactly $b_i-j_i+s_i+b_{i+1}-b_i+j_i-s_i=b_{i+1}$ marked blocks. This proves $\gamma=\gamma_{\sigma_{\Sigma}}(w)$.\qed
\end{proof}

%% file: proof_lengthforebs.tex
\begin{proof}
Let $w\in\mathfrak{V}_{\gamma}$ and $\zeta=(s_1,\dots,s_{\ell-2})$.
By the definition of the $\ebs$ we have $|w|_{\ta_1}=b_1$ and $|w|_{\ta_{\ell}}=b_{\ell-1}-1$. Let $i\in[\ell-1]_{>1}$. By $\iota$ and $\zeta$ we know that we have $j_{i-1}$ join occurrences and $s_{i-1}$ separating occurrences of $\ta_i$. The number of satellites of $\ta_i$ can be calculated in the following way: at stage $i$ there are $b_{i-1} - j_{i-1} + s_{i-1}$ blocks marked considering only the previous blocks as well as the join and separator occurrences of the letter $\ta_i$. Since there must be $b_i$ blocks in the end there have to be exactly $b_i - (b_{i-1} - j_{i-1} + s_{i-1})$ satellites of the letter $\ta_i$. Thus we have $b_i-b_{i-1}+j_{i-1}-s_{i-1}+j_{i-1}+s_{i-1}=b_i-b_{i-1}+2j_{j-1}$ occurrences of $\ta_i$. This concludes the proof.\qed
\end{proof}

%% file: proof_occletter.tex
\begin{proof}
Consider the valid extended blocksequence $\gamma=(\beta,\iota,\zeta)$ with the blocksequence $\beta=(b_1,\dots,b_n)$, the sequence of joins $\iota=(j_1,\dots,j_{n-2})$, and the sequence of separators $\zeta=(s_1,\dots,s_{n-2})$ for $n\in\N_{\geq 2}$. Since $w,w'\in\mathfrak{V}_{\gamma}$ they have exactly $b_1$ (the first element of the block sequence) occurrences of the first letter and 
the same number of join-occurrences and singletons (explicitly given by the number of separators and the next number of blocks) in every marking stage that follows.
Because there are no neighbours and the words are condensed $|w|_x=|w'|_x$ holds for all $x\in\Sigma$ and therefore also $|w| = |w'|$.\qed
\end{proof}

%% file: normalform.tex
As seen in the previous section, adding neighbours does not change the blocksequence. For this reason, we restrict ourselves to words that are neighbourless w.r.t. $\sigma_{\Sigma}$. In this section, we firstly present some results regarding neighbourless marking sequences before we present a normal form for $\mathfrak{V}_{\gamma}$ for a valid $\ebs$ $\gamma$.

\ifpaper
\begin{theorem}[$\ast$]
\else
\begin{theorem}
\fi
\label{orderofletters}
Given a word $w\in\Sigma^n$, a marking sequence $\sigma=(y_1,\dots,y_{\ell})$ is neighbourless for $w$ iff  
$w[1]<_{\sigma}w[2]$, $w[n]<_{\sigma}w[n-1]$, and
for all $i\in[\lfloor\frac{n}{2}\rfloor-1]_{>1}$ we have $w[2i]>_{\sigma}w[2i+1]$ and $w[2i-1]<_{\sigma}w[2i]$.
\end{theorem}
\ifpaper
\else\input{proof_orderofletters}
\fi

\begin{remark}
By Theorem~\ref{orderofletters}, only words of odd length can be neighbourless.
\end{remark}

The following two results are of algorithmic nature We use the standard computational model RAM with logarithmic word-size (see, e.g., \cite{KarkkainenSB06}). Following a standard assumption from stringology (see, e.g., \cite{KarkkainenSB06}), if $w$ is the input word for our algorithms, we can assume that $\Sigma=\letters(w)=\{1,2,\ldots,\ell\}$ with $\ell \leq |w|$. 
Since we restrict ourselves to neighbourless words and marking sequences, we show how to check in linear time whether a word is neighbourless and, if that is the case, how the $\ebs$ can be computed within the same time complexity. 

\ifpaper
\begin{proposition}[$\ast$]
\else
\begin{proposition}
\fi
\label{nlcheck}
We can check whether a condensed word $w\in\Sigma^n$ is neighbourless, and compute a neighbourless marking sequence, in $O(n)$ time.
\end{proposition}

\ifpaper
\else \input{proof_nlcheck}
\fi

\ifpaper
\begin{proposition}[$\ast$]
\else
\begin{proposition}
\fi
\label{ebs}
Given a condensed word $w\in\Sigma^n$ and a neighbourless marking sequence $\sigma$ for $w$, the $\ebs$ of $w$ w.r.t $\sigma$ can be computed in $O(n)$ time.
\end{proposition}

\ifpaper
\else \input{proof_ebs}
\fi

From now on, we will assume that $w\in\Sigma^n$, for $n\in\N$, is neighbourless w.r.t $\sigma_{\Sigma}$. For each valid $\ebs$ 
$\gamma$ we define a normal form $w_{\gamma}\in \mathfrak{V}_{\gamma}$ according to Theorem~\ref{validity} such that $w_{\gamma}$ is neighbourless w.r.t. $\sigma_{\Sigma}$.

\begin{definition}\label{defnormalform}
For a valid $\ebs$ $\gamma=(\beta,\iota,\zeta)$ with $\beta=(b_1,\dots,b_{\ell})$, $\iota=(j_1,\dots$, $j_{{\ell}-2})$, and $\zeta=(s_1,\dots,s_{{\ell}-2})$ define $w_{\gamma}$ by $(v_i)_{i\in[\ell]}$ with $v_{\ell}=w_{\gamma}$ with
$v_1=(\ta_1\bullet)^{b_1}$ and for $i\in[\ell-2]_{>1}$ define $v_{i+1}$ as follows: firstly replace the first $j_i$ occurrences of $\bullet$ by $\ta_{i+1}$, then replace the next $\bullet$ by $(\bullet \ta_{i+1})^{s_i}\bullet$, denote the obtained word by $v$ and set $v_{i+1}$ to $v$ if $b_{i+1}-(b_i-j_i+s_i)=0$ or to $v(\bullet \ta_{i+1})^{b_{i+1}-(b_i-j_i+s_i)}$ otherwise. For obtaining $v_{\ell}$ from $v_{\ell-1}$ replace all occurrences of $\bullet$ by $\ta_n$.
\end{definition}

\begin{remark}
For each word $w\in\Sigma^{\ast}$ which is neighbourless w.r.t $\sigma_{\Sigma}$ we have a corresponding valid $\ebs$ $\gamma_w$. 
Thus, one can define the word $w_{\gamma_w}\in  \mathfrak{V}_{\gamma_w}$, as in Definition \ref{defnormalform}. This word will be called the normal form of $w$ w.r.t. the marking sequence $\sigma_\Sigma$ and the corresponding $\ebs$~$\gamma_w$.
\end{remark}

In the following part, we present three rules with which the normal form of a given $w\in\Sigma^{\ast}$ can be obtained.

\begin{definition}
Let $\gamma$ be a valid $\ebs$ and $i\in[\ell]_{>1}$. Define the following rules:\\
\textbf{Filling the leftmost gaps with joins $R_1$}: For $w=x_1\ta_{k_1}u\ta_{k_2}v\ta_{k_3}x_2$ with $u,v\in\Sigma^+$, $k_1,k_2,k_3<i$ and $x_1,x_2\in(\Sigma\cup\overline{\Sigma})^{\ast}$ define the application of $R_1$ by $w'=x_1\ta_{k_1}v\ta_{k_2}u\ta_{k_3}x_2$ (see Fig.~\ref{figr1}).\\
\begin{figure}[h]
\centering
\input{r1.tex}
\caption{Appl. of $R_1$: dark is marked, light is unmarked, shaded contains both kinds.}\label{figr1}
\end{figure}
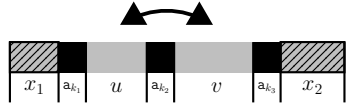
\noindent\textbf{All separators in one gap $R_2$}: Consider $w=x_1z_1u\overrightarrow{\ta_i}^{k_1}z_2v
\overrightarrow{\ta_i}^{k_2}z_3x_2$ with
$x_1,x_2\in(\Sigma\cup\overline{\Sigma})^{\ast}$, $u,v\in\{\ta_{i+1},\dots,\ta_n\}^+$,  $z_1,z_3\in\{\ta_1,\dots,\ta_{i-1}\}$, $z_2\in(\Sigma\cup\overline{\Sigma})^+$ with $z_2[1],z_2[|z_2|]\in\overline{\Sigma}$ and with $\overrightarrow{\ta_i}^{k_1}=\ta_i t_1\ta_i t_2\dots \ta_i t_{k_1}$ and $\overrightarrow{\ta_i}^{k_2}=\ta_i t_1'\ta_it_2'\dots \ta_it_{k_2}'$. For an application of $R_2$ chose $m_1,m_2\in\N$ with $m_1+m_2=k_1+k_2$ and $r_1,\dots,r_{m_1}\in\{t_1,\dots,t_{k_1},t_1',\dots,t_{k_2}'\}$ different as well as $r_1',\dots,r_{m_2}'$ with $\{r_1',\dots,r_{m_2}'\}=\{t_1,\dots,t_{k_1},t_1',\dots,t_{k_2}'\}\backslash\{r_1,\dots,r_{m_1}\}$. Set $p=\ta_ir_1\ta_ir_2\dots \ta_ir_{m_1}$ and $p'=\ta_ir_1'\ta_ir_2'\dots \ta_ir_{m_2}'$. Then the application of $R_2$ to $w$ results in $w'=x_1z_1upz_2vp'z_3x_2$ (see Fig.~\ref{figr2}).\\
\begin{figure}[h]
\centering
\input{r2.tex}
\caption{Appl. of $R_2$: dark is marked, light is unmarked, shaded contains both kinds.}\label{figr2}
\end{figure}
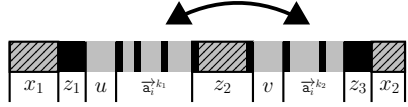
\noindent\textbf{Moving satellites $R_3$}: For $w=x_1\overrightarrow{\ta_i}^r c_i \overleftarrow{\ta_i}^s x_2$ with $r,s \in\N_0$, $x_1,x_2\in(\Sigma\cup\overline{\Sigma})^{\ast}$ and the core $c_i$ at stage $i$, define the application of 
$R_3$(a) by $w'=x_1     (\overleftarrow{\ta_i}^s)^R \overrightarrow{\ta_i}^r     c_ix_2$ and of 
$R_3$(b) by $w'=x_1 c_i \overleftarrow{\ta_i}^s     (\overrightarrow{\ta_i}^r)^R    x_2$ 
(see Fig.~\ref{fig3}).
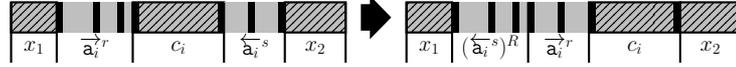
\begin{figure}[h]
\centering
\input{r3.tex}
\caption{Appl. of $R_3$: dark is marked, light is unmarked, shaded contains both kinds.}\label{fig3}
\end{figure}
\end{definition}

\ifpaper
\begin{theorem}[$\ast$]
\else
\begin{theorem}
\fi
\label{obtainingnormalform}
For all $w\in\Sigma^{\ast}$ there exists a sequence $(r_1,\dots,r_m)$ with $r_i\in\{R_1,R_2,R_3\}$, $i\in[m]$, $m\in\N_0$ such that $w_{\gamma_w}$ is obtained from $w$ w.r.t. $\sigma_{\Sigma}$.
\end{theorem}

\ifpaper
\else \input{proof_obtainingnormalform}
\fi

\ifpaper
\begin{lemma}[$\ast$]
\else
\begin{lemma}
\fi
\label{class}
For a valid $\ebs$ $\gamma$ and $w \in \mathfrak{V}_{\gamma}$ applying any one of the rules $R_1$, $R_2$, or $R_3$ to $w$ resulting in the word $w'$ we get $w' \in \mathfrak{V}_{\gamma}$ as well.
\end{lemma}

\ifpaper
\else \input{proof_class}
\fi

\ifpaper
\begin{corollary}[$\ast$]
\else
\begin{corollary}
\fi
\label{charVgamma}
For a given valid $\ebs$, $\mathfrak{V}_{\gamma}$ contains exactly one normal form and all words having this normal form are in $\mathfrak{V}_{\gamma}$.
\end{corollary}

\ifpaper
\else \input{proof_charVgamma}
\fi

For the $\ebs$ $\gamma = ((4,4,1),(1),(1))$, we have $w_{\gamma}=\ta\tb\ta\tc\tb\tc\ta\tc\ta$ and $\mathfrak{V}_{\gamma} = \{\ta\tb\ta\tc\ta\tc\tb\tc\ta,\ta\tb\ta\tc\tb\tc\ta\tc\ta,\ta\tc\ta\tb\ta\tc\tb\tc\ta,\ta\tc\tb\tc\ta\tb\ta\tc\ta,\ta\tc\ta\tc\tb\tc\ta\tb\ta,\ta\tc\tb\tc\ta\tc\ta\tb\ta\}$ assuming $\sigma_{\Sigma}$ as marking sequence.

%% file: proof_orderofletters.tex
\begin{proof}
Let first $\sigma$ be a neighbourless marking sequence for $w$. If $w[2]$ would be marked before $w[1]$ we had with $w[1]$ a neighbouring occurrence. Analogously we have $w[n]<_{\sigma}w[n-1]$.  Suppose that there exists an $i\in[\lfloor\frac{n}{2}-1\rfloor]_{>1}$ with $w[2i]<_{\sigma} w[2i+1]$ or $w[2i-1]>_{\sigma} w[2i]$ (equality is excluded since $w$ is condensed). Choose $i$ minimal. We have to consider three cases.\\
\textbf{case 1:} $w[2i]<_{\sigma} w[2i+1]$ and $w[2i-1]>_{\sigma} w[2i]$.\\
In this case the letters $w[2i-3],w[2i-2],w[2i-1],w[2i]$ are of interest. We know that $w[2i-3]$ is marked before $w[2i-2]$ by the minimality of $i$ and $w[2i]$ is marked before $w[2i-1]$ which is marked before $w[2i-2]$ by the case-constraint. This implies that either $w[2i-2]$ is marked while $w[2i-1]$ is unmarked and thus $w[2i-2]$ is a neighbour or vice versa.\\
\textbf{case 2:} $w[2i]<_{\sigma} w[2i+1]$ and $w[2i-1]<_{\sigma} w[2i]$ \\
In this case $w[2i]$ is marked when $w[2i+1]$ ($w[2i-1]$) is already marked and $w[2i-1]$ ($w[2i+1]$) is still unmarked and thus $w[2i]$ is a neighbouring occurrence.\\
\textbf{case 3:} $w[2i]\geq_{\sigma} w[2i+1]$ and $w[2i-1]>_{\sigma} w[2i]$\\
This case is similar to case 2.\\
Since we get a contradiction in all cases, the $\Rightarrow$-direction is proven.

\medskip

Assume for the other direction that the constraints hold. Let $j\in[n]$. The constraints ensure that $w[j-1]$ and $w[j+1]$ are either both marked before $w[j]$ and thus $w[j]$ is join occurrence or both are marked after $w[j]$ and thus $w[j]$ is a separator. Hence, $\sigma$ is neighbourless.\qed
\end{proof}

%% file: proof_nlcheck.tex
\begin{proof}
We can assume $n=|w|$ is odd. Firstly we build a directed graph based on Theorem~\ref{orderofletters}.
Define $G_w=(\Sigma,E)$ with $E\subset\{(\ta,\tb)|\,\ta,\tb\in\Sigma\}$ as follows. Firstly, add the directed edges $(w[1],w[2])$ and $(w[n],w[n-1])$to $E$. 
Then, for all $i\in[\lfloor\frac{n}{2}\rfloor-1]_{>1}$ we add the edges $(w[2i+1],w[2i])$ and $(w[2i-1],w[2i])$.
Intuitively, we have an edge $(\ta,\tb)$ if and only if we would need to have $\ta <_\sigma \tb$ in any neighbourless marking sequence $\sigma$ for $w$. 

To find such a neighbourless marking sequence $\sigma$, it is enough to find the linear ordering of the vertices of $V$ such that for every directed edge $(\ta,\tb)$ from vertex $\ta$ to vertex $\tb$, $\ta$ comes before $\tb$ in the respective ordering. Such a sequence can be found using a standard topological sorting algorithm based on the depth-first search (DFS). Such an algorithm produces successfully a linear ordering of the vertices of $G_w$ (and, as such, a neighbourless marking sequence $\sigma$ for $w$) if and only if $G_w$ is a directed acyclic graph (DAG). 
The time complexity of this algorithm is $O(\ell+|E|)=O(n)$.\qed
\end{proof}

%% file: proof_ebs.tex
\begin{proof}
Recall that, when discussing algorithms, we work under the assumption that $\Sigma= \{1,2,\ldots,\ell\}$ for some $\ell \leq |w|$. Moreover we can assume w.l.o.g. $\sigma=\sigma_\Sigma$. Otherwise, we could rename the letters of $w$ by replacing $\ta$ by $\sigma^{-1}(\ta)$  ($\sigma$ is a permutation of $\Sigma$, so we can invert it) for all $\ta\in \Sigma$. This means that if $\ta $ is the \nth{i} in the ordering $\sigma$, we replace $\ta$ by $i$. In this way, we obtain a word $w'$ from $w$, which is neighbourless w.r.t. $\sigma_\Sigma$, given that $w$ was neighbourless w.r.t. $\sigma$. 
 
So, in the following, $w \in \Sigma^n$ is a neighbourless word w.r.t. $\sigma=\sigma_\Sigma$, and we want to produce the $\ebs$ of $w$ w.r.t. $\sigma$. 
Firstly, we construct a list-array $\Pos$ of size $\ell$, where $\Pos[\ta]$ stores the list of the positions (ordered from left to right) where $\ta$ occurs in $w$, for all $\ta\in \Sigma$.
Further, we construct the arrays $F$ and $L$ of size $\ell$ to store the first and last occurrence of every letter in $w$. More precisely, $F[\ta]=\min\{i\in [n]\mid w[i]=\ta\}$ and $L[\ta]=\max\{i\in [n]\mid w[i]=\ta\}$, for all $\ta\in \Sigma$. 
Clearly, all these arrays can be computed by going once through $w$. 

Now, we compute the arrays $F'$ and $L'$, both of size $\ell -1$, where $F'[i] = \min(F[1 \ldots (i-1)])$ and $L'[i] = \max(L[1 \ldots (i-1)])$ for $2 \leq i \leq \ell$.
The entry $F'[i]$ represents the first position in $w$ where a letter strictly smaller than $i$ occurs, while $L'[i]$ is the last position where a letter strictly smaller than $i$ occurs. Note that all letters strictly smaller than $i$ are marked before $i$ in our $\sigma$. The arrays $F'$ and $L'$ can be computed by a simple dynamic programming approach in linear time. For instance, $F'$ is computed using the formula $F[1]=n+1$ and $F'[i]=\min\{F'[i-1],F[i-1]\}$ for $i\geq 2$, and $L'$ is computed using the formula $L[1]=0$ and $L'[i]=\max\{L'[i-1],L[i-1]\}$ for $i\geq 2$. Intuitively, before executing the \nth{i} marking step in the sequence $\sigma$, $F'[i]$ is the leftmost marked symbol and $L'[i]$ the rightmost marked symbol. 

Now, we start marking the positions of $w$ following $\sigma$. We first mark the positions in $\Pos[1]$, and count the number of blocks we obtain this way. Further, we explain how we mark the positions in $\Pos[i]$, for $i\geq 2$. We go through the list $\Pos[i]$ left to right. If $j$ is the current element (i.e., position of $w$), we check whether $w[j-1]$ and $w[j]$ are marked. If yes $j$ is a join occurrence. If none of $w[j-1]$ and $w[j]$ are marked, and $F'[i]<j<L'[i]$ (meaning that $w[j]$ occurs between two already marked blocks), then the position $j$ is a separator. Finally, the position $j$ is a satellite, otherwise (as $\sigma$ is neighbourless).

The join sequence $\iota$, the separator sequence $\zeta$, as well as the number of satellite occurrences of each letter, can be calculated directly from the procedure described above. Everything takes clearly $O(n)$ time.

The blocksequence $\beta = (b_1, \cdots,b_{\ell-1} ,1)$ can be calculated from the join and separator sequences with the additional information about the satellite occurrences, similar to Theorem~\ref{lengthforebs}: $b_1$ is the number of occurrences of $1$ and $b_i=b_{i-1} - \iota_{i-1} + \zeta_{i-1} + \zeta'_{i-1}$, where $\zeta'_{i-1}$ is the number of satellites of $i$. The calculations can be done in linear time $O(n+\ell)$. This leads to a total time for the algorithm in $O(n)$.\qed
\end{proof}

%% file: r1.tex
\begin{tikzpicture}[font=\sffamily, thick, scale=0.4, every node/.style={transform shape}]
\fill[lightgray](-1,0) rectangle (0.6,1);
\draw[black, pattern=north east lines, pattern color=black] (-1,0) rectangle (0.6,1);
\node at (-0.2,-0.5) {\huge $x_1$};
\draw (-1,1) -- (-1,-1);
\fill(0.6,0) rectangle (1.5,1);
\node at (1.05,-0.5) {\Large $\ta_{k_1}$};
\draw (0.6,1) -- (0.6,-1);
\fill[lightgray](1.5,0) rectangle (3.5,1);
\node at (2.5,-0.5) {\huge $u$};
\draw (1.5,1) -- (1.5,-1);

\fill(3.5,0) rectangle (4.4,1);
\node at (3.95,-0.5) {\Large $\ta_{k_2}$};
\draw (3.5,1) -- (3.5,-1);

\fill[lightgray](4.4,0) rectangle (7,1);
\node at (5.8,-0.5) {\huge $v$};
\draw (4.4,1) -- (4.4,-1);

\fill(7,0) rectangle (7.9,1);
\node at (7.45,-0.5) {\Large $\ta_{k_3}$};
\draw (7,1) -- (7,-1);

\fill[lightgray](7.9,0) rectangle (10,1);
\draw[black, pattern=north east lines, pattern color=black] (7.9,0) rectangle (10,1);
\draw (7.9,1) -- (7.9,-1);
\node at (8.95,-0.5) {\huge $x_2$};
\draw (10,1) -- (10,-1);

\node (1) at (2.5,1.5) {};
\node (2) at (5.8,1.5) {};
 \path[triangle 60-triangle 60,line width=0.5mm]
    (2) edge[bend right] node [left] {} (1);

 \end{tikzpicture}

%% file: r2.tex
\begin{tikzpicture}[font=\sffamily, thick, scale=0.4, every node/.style={transform shape}]

\fill[lightgray](-1,0) rectangle (0.6,1);
\draw[black, pattern=north east lines, pattern color=black] (-1,0) rectangle (0.6,1);
\node at (-0.2,-0.5) {\huge $x_1$};
\draw (-1,1) -- (-1,-1);
\fill(0.6,0) rectangle (1.5,1);
\node at (1.05,-0.5) {\huge $z_1$};
\draw (0.6,1) -- (0.6,-1);
\fill[lightgray](1.5,0) rectangle (2.5,1);
\node at (2,-0.5) {\huge $u$};
\draw (1.5,1) -- (1.5,-1);
\fill[lightgray](2.5,0) rectangle (5,1);
\fill(2.5,0) rectangle (2.7,1);
\fill(3.1,0) rectangle (3.3,1);
\fill(4,0) rectangle (4.2,1);
\draw (2.5,1) -- (2.5,-1);
\node at (3.75,-0.5) {\Large $\overrightarrow{\ta_i}^{k_1}$};
\fill[lightgray](5,0) rectangle (7,1);
\draw[black, pattern=north east lines, pattern color=black] (5,0) rectangle (7,1);
\fill[black](5,0) rectangle (5.2,1);
\fill[black](6.8,0) rectangle (7,1);
\draw (5,1) -- (5,-1);
\node at (6,-0.5) {\huge $z_2$};
\fill[lightgray](7,0) rectangle (8,1);
\draw (7,1) -- (7,-1);
\node at (7.5,-0.5) {\huge $v$};
\fill[lightgray](8,0) rectangle (10,1);
\fill(8,0) rectangle (8.2,1);
\fill(9.2,0) rectangle (9.4,1);
\draw (8,1) -- (8,-1);
\node at (9,-0.5) {\Large $\overrightarrow{\ta_i}^{k_2}$};
\fill(10,0) rectangle (10.9,1);
\draw (10,1) -- (10,-1);
\node at (10.45,-0.5) {\huge $z_3$};
\fill[lightgray](10.9,0) rectangle (12,1);
\draw[black, pattern=north east lines, pattern color=black] (10.9,0) rectangle (12,1);
\draw (10.9,1) -- (10.9,-1);
\node at (11.5,-0.5) {\huge $x_2$};
\draw (12,1) -- (12,-1);

\node (1) at (4,1.5) {};
\node (2) at (9,1.5) {};
 \path[triangle 60-triangle 60,line width=0.5mm]
    (2) edge[bend right] node [left] {} (1);

 \end{tikzpicture}

%% file: r3.tex
\begin{tikzpicture}[font=\sffamily, thick, scale=0.4, every node/.style={transform shape}]

\fill[lightgray](0,5) rectangle (1.5,6);
\draw[black, pattern=north east lines, pattern color=black] (0,5) rectangle (1.5,6);
\node at (0.75,4.5) {\huge $x_1$};
\draw (0,6) -- (0,4);

\fill[lightgray](1.5,5) rectangle (5,6);
\fill(1.5,5) rectangle (1.7,6);
\fill(2.7,5) rectangle (2.9,6);
\fill(3.5,5) rectangle (3.7,6);
\node at (2.75,4.5) {\huge $\overrightarrow{\ta_i}^r$};
\draw (1.5,6) -- (1.5,4);
%
\fill[lightgray](4,5) rectangle (7,6);
\draw[black, pattern=north east lines, pattern color=black] (4,5) rectangle (7,6);
\fill(4,5) rectangle (4.2,6);
\node at (5.5,4.5) {\huge $c_i$};
\fill(6.8,5) rectangle (7,6);
\draw (4,6) -- (4,4);

\fill[lightgray](7,5) rectangle (9,6);
\fill(7.8,5) rectangle (8,6);
\fill(8.8,5) rectangle (9,6);
\node at (8,4.5) {\huge $\overleftarrow{\ta_i}^s$};
\draw (7,6) -- (7,4);

\fill[lightgray](9,5) rectangle (11,6);
\draw[black, pattern=north east lines, pattern color=black] (9,5) rectangle (11,6);
\node at (10,4.5) {\huge $x_2$};
\draw (9,6) -- (9,4);
\draw (11,6) -- (11,4);
\fill[lightgray](13,5) rectangle (14.5,6);
\draw[black, pattern=north east lines, pattern color=black] (13,5) rectangle (14.5,6);
\node at (13.75,4.5) {\huge $x_1$};
\draw (13,6) -- (13,4);

\fill[lightgray](14.5,5) rectangle (17,6);
\fill(14.5,5) rectangle (14.7,6);
\fill(15.7,5) rectangle (15.9,6);
\fill(16.5,5) rectangle (16.7,6);
\node at (15.75,4.5) {\huge $(\overleftarrow{\ta_i}^s)^R$};
\draw (14.5,6) -- (14.5,4);

\fill[lightgray](17,5) rectangle (19,6);
\fill(17,5) rectangle (17.2,6);
\fill(18,5) rectangle (18.2,6);
\node at (18,4.5) {\huge $\overrightarrow{\ta_i}^r$};
\draw (17,6) -- (17,4);

\fill[lightgray](19,5) rectangle (22,6);
\draw[black, pattern=north east lines, pattern color=black] (19,5) rectangle (22,6);
\fill(19,5) rectangle (19.2,6);
\node at (20.5,4.5) {\huge $c_i$};
\fill(21.8,5) rectangle (22,6);
\draw (19,6) -- (19,4);

\fill[lightgray](22,5) rectangle (24,6);
\draw[black, pattern=north east lines, pattern color=black] (22,5) rectangle (24,6);
\node at (23,4.5) {\huge $x_2$};
\draw (22,6) -- (22,4);
\draw (24,6) -- (24,4);

\draw[-{Triangle[width=16pt,length=8pt]}, line width=10pt](11.5,5.5) -- (12.5, 5.5);

 \end{tikzpicture}

%% file: proof_obtainingnormalform.tex
\begin{proof}
We construct the normal form inductively for the extended blocksequence corresponding to $w$ marked with $\sigma_{\Sigma}$. Perform the following four steps for all $i\in[\ell]_{>1}$ (cf. Definition~\ref{rules}):
\begin{enumerate}
\item as long as $w_i$ can be written as $x_1\ta_{k_1}u\ta_{k_2}\ta_i\ta_{k_3}x_2$ with $u\neq \ta_i$ apply $R_1$ (fill the first gaps with join occurrences),
\item if $w_i$ can be written as $x_1\ta_{k_1}u\ta_{k_2}v\ta_{k_3}x_2$ with $v\neq \ta_i$, $\ta_i\in\letters(v)$ and $u\neq \ta_i$, apply $R_1$ (move one block containing a separator occurrence to the gap immediately right to the gaps filled with joins),
\item if $w_i$ can be written as $x_1z_1u\overrightarrow{\ta_i}^{k_1}z_2v\overrightarrow{\ta_i}^{k_2}
z_3x_2$, apply $R_2$ with $m_1=k_1+k_2$ and $m_2=0$ (move all separators into the same gap),
\item if $w_i$ can be written as $x_1\overrightarrow{\ta_i}^r c_i\overleftarrow{\ta_i}^s x_2$, apply $R_3(b)$ (move satellites to the right).
\end{enumerate}
We prove by induction on $i\in[\ell]_{\geq 2}$ that after each application of all four steps the join occurrences of $\ta_i$ are in the leftmost gaps, in the following gap are the separating occurrences of $\ta_i$, and all satellite occurrences of $\ta_i$ are to the right of all $\ta_{i-1}$.
Let $i=2$. Then $w_i$ is of the form $u_1\ta_1u_2\ta_1u_3\dots u_{b_1}\ta_1u_{b_1+1}$ with $u_i\in(\Sigma\backslash\{\ta_1\}))^{\ast}$ for $i\in[b_1+1]$. By $\gamma_w$ we know that $w_i$ has $j_1$ join occurrences of $\ta_2$ and thus there exists exactly $j_1$ $u_i$ with $u_i=\ta_2$. Since we apply $R_1$ as long as there exists a join-occurrence that has a factor $\ta_1u_k\ta_1$ to the left in $w_i$, we obtain after the first step the word $u_1(\ta_1\ta_2)^{j_1}\ta_1xu_{b_1+1}$ with $x=\varepsilon$ if $j_1=b_1-1$ and $x=u_{b_1-j_1-2}\ta_1u_{b_1-j_1-1}\dots u_{b_1}\ta_1$ otherwise. If $s_1=0$ the second step is skipped. Assume $s_1>0$. Thus, there exist separating occurrences of $\ta_2$ and hence at least one $u_i$ is of the form $x_1\ta_2x_2$ with $\letters(x_1x_2)\subseteq\Sigma\backslash\{\ta_1,\ta_2\}$. If this $u_i$ is not $u_{b_1-j_1-2}$, $R_1$ is applied such that $u_{b_1-j_1-2}$ and $u_i$ switch positions. This application results in the word 
\[
u_1(\ta_1\ta_2)^{j_1}\ta_1\cdot x_1\ta_2x_2 \ta_1 u_{b_1-j_1-1}\ta_1 \dots  \ta_1 u_{b_1-j_1-2}\ta_1u_{b_1-j_1-1}\dots u_{b_1}\ta_1\cdot u_{b_1+1}.
\]
In the third step all other separating occurrences of $\ta_2$ are moved by rule $R_2$ into the same gap,
i.e. if there exists another $u_{i'}=x_1'\ta_2x_2'$ with $\letters(x_1',x_2')=\Sigma\backslash
\{\ta_1,\ta_2\}$ the application of $R_2$ leads to
\begin{multline*}
u_1(\ta_1\ta_2)^{j_1}\ta_1\cdot x_1\ta_2x_2\ta_2x_2' \ta_1 u_{b_1-j_1-1} \dots  \ta_1 u_{b_1-j_1-2}\ta_1u_{b_1-j_1-1}\\
\dots
\ta_1 x_1' \ta_1 \dots u_{b_1}\ta_1\cdot    u_{b_1+1}
\end{multline*}
and finally to
\[
u_1(\ta_1\ta_2)^{j_1}\ta_1x_1\overrightarrow{\ta_2}^{s_1}z_1\ta_1z_2\dots z_{\ell}\ta_1z_{\ell+1}
\]
for appropriate $z_i\in(\Sigma\backslash\{\ta_1,\ta_2\})^{\ast}$, $i\in[\ell]$, for some $\ell\in\N$, and $z_{\ell+1}\in(\Sigma\backslash\{\ta_1\})^{\ast}$.
After this step all separating occurrences of $\ta_2$ are between two occurrences of $\ta_1$ and especially between those $\ta_1$s such that all previous occurrences are directly joined by one occurrence of $\ta_2$. Thus only $u_1$ and $z_{\ell+1}$ may contain occurrences of $\ta_2$ (in the form of satellite occurrences). Set $c_2$ as the factor starting at $|u_1|+1$ and ending just before $z_{\ell+1}$. If $w_i$ can be written as $x_1 \overrightarrow{\ta_i}^r c_1 \overleftarrow{\ta_i}^s x_2$ as in rule $R_3$ apply $R_3(b)$. After this step there is no occurrences of $\ta_2$ before the first $\ta_1$ and all satellite occurrences of $\ta_2$ are after the last 
$\ta_1$.

\medskip
Consider now the \nth{$i$} stage. Then we have a word of the form 
\[
u_1B_1u_2B_2u_3\dots u_{b_{i-1}}B_{b_{i-1}}u_{b_{i-1}+1}
\]
with $B_i\in\overline{\Sigma}^{\ast}$ and $u_j\in\Sigma^{\ast}$ for $i\in[b_{i-1}]$ and $j\in[b_{i-1}+1]$. Since there exist $j_{i-1}$ join occurrences of $\ta_i$, there exist exactly $j_{i-1}$ occurrences of $u_k$ which are equal to $\ta_i$. Applying $j_{i-1}$ times rule $R_1$ we swap these occurrences with $u_2,\dots,u_{j_{i-1}+1}$. In the next step we swap a separating occurrences of $\ta_{i}$ with $u_{j_{i-1}+2}$ (if one exists). After these transformation we have with appropriate $z_1,\dots,z_{\ell}$ a word of the form
\[
u_1B_1\ta_iB_2\ta_i\dots \ta_i B_{j_{i-1}+1} z_1 B_{j_{i-1}+2} z_2 \dots z_{\ell-1}B_{b_{i-1}}z_{\ell}
\]
such that $z_1$ contains all separating occurrences of $y_i$. Now define $c_i$ as the factor starting right after $u_1$ and ending just before $z_{\ell}$. 

Notice that $u_1$ and $z_{\ell}$ only may contain satellite occurrences of $y_i$ and $c_i$ does not contain any. With $R_3(b)$ we move these occurrences all to the right into $z_{\ell}$. This proves that the application of the four steps results in a word where always the first gaps are filled with the join-occurrences, followed by a gap containing all separating occurrences, and that the satellite occurrences are all at the right side of the core in each marking step.

\medskip

Comparing this word with the definition of $w_{\gamma_w}$ leads to the claim.\qed
\end{proof}

%% file: proof_class.tex
\begin{proof}
Let $\gamma=(\beta,\iota,\zeta)$ be the valid extended blocksequence w.r.t. $\sigma_{\Sigma}$, blocksequence $\beta=(b_1, \dots b_n)$, join sequence $\iota=(j_1, \dots j_{n-2})$, and separator sequence $\zeta=(s_1, \dots s_{n-2})$. We divide the proof into three parts, one for each rule. Consider the stage $i\in[n]_{>1}$.\\
\textbf{case $R_1$:} If we apply $R_1$ to $w$ resulting in $w'$, these words are of the form
\[
w=x_1\ta_{k_1}u\ta_{k_2}v\ta_{k_3}x_2\mbox{ and }w'=x_1\ta_{k_1}v\ta_{k_2}u\ta_{k_3}x_2
\]
with $u,v\in\Sigma^+$, $k_1,k_2,k_3<i$ and $x_1,x_2\in(\Sigma\cup\overline{\Sigma})^{\ast}$. 
Since $k_1,k_2,k_3<i$ the letters $\ta_{k_1}, \ta_{k_2}$, and $\ta_{k_3}$ are all marked at stage $i$. Notice that $x_1 \ta_{k_1}, \ta_{k_2}$ and $\ta_{k_3}x_2$ are factors of both $w$ and $w'$. Thus, they have the same number of blocks, join occurrences and separating occurrences in $w$ and $w'$. 
Since $u$ and $v$ are surrounded by $\ta_{k_1}, \ta_{k_2}$ and $\ta_{k_3}$ which are all marked, there cannot be any satellite occurrences of $\ta_i$ in $u$ or $v$.
If $u$ is a join occurrence then it is of the form $u=\ta_i$. In this case it remains a join in between $\ta_{k_2}$ and $\ta_{k_3}x_2$.
On the other hand if u contains any number of separating occurrences it is of the form $u=u_1 \ta_i u_2 \ta_i \dots \ta_i u_m$ with $u_1, \dots, u_m \in \Sigma^+$. Specifically $u_1, u_m$ are not marked in this step and not empty. Therefore, all separating occurrences of $\ta_i$ in $u$ remain to be ones in between $\ta_{k_2}$ and $\ta_{k_3}x_2$.
The same holds analogously for v. Therefore, the number of joins and separators is the same in $w$ and $w'$ in stage i of the marking. Since there are no neighbours or satellites in $u$ and $v$ the number of blocks at stage $i$ is the same as well.\\
\textbf{case $R_2$:} If we apply $R_2$ to $w$ resulting in $w'$ these words are of the form 
\[
w=x_1z_1u\overrightarrow{\ta_i}^{k_1}z_2v \overrightarrow{\ta_i}^{k_2}z_3x_2\mbox{ and }
w'=x_1z_1u\overrightarrow{\ta_i}^{k_1}\overrightarrow{\ta_i}^{k_2}z_2vp'z_3x_2
\]
with $k_1,k_2\in\N$.
In order to show that $w' \in \mathfrak{V}_{\gamma_w}$ holds, we will show that the amounts of satellites, separators and joins for both words are the same and therefore both words have the same extended blocksequence.
The factors $x_1z_1u$, $z_2v$, and $z_3x_2$ are the same in both words. The letters that are adjacent to those factors can differ in $w$ and $w'$. However, the order in which a letter adjacent to one these factors and the letter on the border of the respective factor are marked, is the same in both words. So the factors $x_1z_1u$, $z_2v$, and $z_3x_2$ behave the same way in $w$ and $w'$. This means only the factors containing $\overrightarrow{\ta_i}$ can cause any change in the extended blocksequence.   
Note that when marking $w$ and $w'$, before reaching stage $i$ both words behave in the same way since 
$\overrightarrow{\ta_i}^{k_1}$ and $\overrightarrow{\ta_i}^{k_2}$ contain only letters that are marked in stage $i$ or later. Satellites are not affected by applying $R_2$ since $\overrightarrow{\ta_i}^{k_1}$ and 
$\overrightarrow{\ta_i}^{k_2}$ lie within the core of $w$ and satellites only occur outside the core.
In stage $i$, $\overrightarrow{\ta_i}^{k_1}$ and $\overrightarrow{\ta_i}^{k_2}$ contain only separators. By the definition of $\overrightarrow{\ta_i}$ and because $k_1+k_2 = m_1+m_2$ holds, the amount of separators is the the same in $w$ and $w'$.
Thus, in $\overrightarrow{\ta_i}^{k_1}$, $\overrightarrow{\ta_i}^{k_2}$, $p$ and $p'$ only $\ta_i$ is marked, resulting in factors that are repetitions of $\ta_i$ (which is marked) followed by an unmarked factor. These unmarked factors occur in different places in $w$ and $w'$ but are encased by either $\ta_i$, $z_2[1]$, or $z_3[1]$ which are all marked at this stage. Hence, we can apply $R_2$ and the number of joins and separators does not change. This concludes the part for $R_2$.
\\
\textbf{case $R_3$:} If we apply $R_3(a)$ (the case $R_3(b)$ works analogously) to $w$ resulting in $w'$ these words are of the form
\[
w=x_1\overrightarrow{\ta_i}^r c_i\overleftarrow{\ta_i}^s x_2\mbox{ and }
w'=x_1 (\overleftarrow{\ta_i}^s)^R \overrightarrow{\ta_i}^r c_ix_2.
\] 
Notice that similar to $R_2$ in $\overrightarrow{\ta_i}^r$ and $\overleftarrow{\ta_i}^s$ letters are firstly marked in stage $i$. So in $w$, $c_i$ and $x_2$ are separated by an unmarked factor, whereas in $w'$ $c_i$ and $x_2$ occur directly next to each other. Notice that $x_1[|x_1|]$ and $x_2[1]$ are unmarked (or empty) since otherwise the words would not be neighbourless.
Since every occurrence of $\ta_i$ in $\overleftarrow{\ta_i}^s$ is neighboured only by unmarked letters in $w$ and $w'$, these occurrences are a single block in both words. The amount of blocks outside $\overleftarrow{\ta_i}^s$ stays unchanged, so it holds that $b_i = b_i'$. 
After stage $i$ the remaining letters in $\overleftarrow{\ta_i}^s$ greater than $\ta_i$ are yet to be marked. Since $\overleftarrow{\ta_i}^s$ is of the form $x_1 \ta_i x_2 \ta_i \cdots x_s \ta_i$ with $x_j \in \{ \ta_{i+1}, \cdots, \ta_l \}^+ $ and every such $x_j$ is encased by letters marked before it in $w$ ($\ta_i$ from the previous step or the last letter of the core) for $j \in [s]$, the same holds for $(\overleftarrow{\ta_i}^s)^R$ in $w'$, by $R_3$ these factors behave the same in both words. Now we have to consider $c_i$ and $x_2$ which become neighbours, when $\overleftarrow{\ta_i}^s$ is moved by the application of $R_3$. If $x_2=\varepsilon$ we are done, otherwise the right neighbour of $c_i$ is an unmarked letter in $w$, the left neighbour of $x_2$ is marked. In $w'$ $c_i$ and $x_2$ occur directly next to each other so $c_i$ also has an unmarked neighbour and $x_2$ also has a marked left neighbour. This means that moving $\overleftarrow{\ta_i}^s$ does not affect the number of joins, separators and satellites, so both words have the same extended blocksequence.\qed
\end{proof}

%% file: proof_charVgamma.tex
\begin{proof}
Let $w,w'\in\mathfrak{V}_{\gamma}$. Then we have $\gamma_w=\gamma=\gamma_{w'}$. By Theorem~\ref{obtainingnormalform} we get that $w,w'$ have the same normal form since the procedure only takes $\gamma$ into account. Let $w_{\gamma}$ be the normal form of $\gamma$. If there existed a $v\in\Sigma^{\ast}$ with normal form $w_{\gamma}$ and 
$v\not\in\mathfrak{V}_{\gamma}$ then $\gamma\neq \gamma_v$. If the differences were in $\iota$ or $\zeta$ the normal form of $\gamma_v$ would have a different amount of joining blocks are a different amount of separators, respectively. If the difference were in the block sequence the normal form would have a different amounts of at least one letter. Thus, $\mathfrak{V}_{\gamma}$ contains exactly the words having the same normal form.\qed
\end{proof}

%% file: localityofnormalform.tex
In the remaining part of this section we investigate the behaviour of a word in comparison to its normal form regarding the locality of the words. 
\ifpaper
\else
For convenience in the following proofs, we introduce the following notions.

\begin{definition}
Let $n(w,S)$ denote the number of marked blocks in $w$ if all letters from $S\subseteq\Sigma$ are 
marked and define for a given $I\subseteq \Sigma$ the function $o$ by $o(\ta,\tb,I)=1$ if $\ta,\tb\in I$ and $0$ otherwise.
\end{definition}
\fi

\ifpaper
\begin{theorem}[$\ast$]
\else
\begin{theorem}
\fi
\label{applR1}
Let $R_i(w)$ denote the application of $R_i$ to $w$ for $i\in[3]$. Then we have that $\loc(R_1(w))$ differs
from $\loc(w)$ by at most $2$ and $\loc(R_2(w))$ and $\loc(R_3(w))$ differ from $\loc(w)$ by at most $1$.
\end{theorem}
\ifpaper
\else \input{proof_applR1}
\fi

For transforming a neighbourless word $w\in\Sigma^{\ast}$ into $w_{\gamma}$ given the $\ebs$ $\gamma$ w.r.t. $\sigma_{\Sigma}$, $R_1$ needs to be applied $\leq j_i+1$ times (moving $j_1$ joins and one separator), $R_2$ $\leq s_i$ times and $R_3$ once (moving all satellites to the right side of the core).

\ifpaper
\begin{corollary}[$\ast$]
\else
\begin{corollary}
\fi
\label{multappl}
For $w\in\Sigma^{\ast}$ and the $\ebs$ $\gamma_w$ induced by $\sigma_{\Sigma}$, we have $\loc(w)\leq\loc(w_{\gamma_w})+\sum_{i\in[\ell]}(2j_i+s_i)+\ell$.
\end{corollary}
\ifpaper
\else \input{proof_multappl}

\fi

In this section, we have proven that for neighbourless words (w.r.t. $\sigma_{\Sigma}$) we can always find a normal form and we showed how the locality of the word itself and its normal form differ in the worst case. This upper bound proven in Corollary~\ref{multappl} can only be reached if at any stage the {\em critical letters}, the letters adjacent to the factors moved by the rules, are all different. Since, for instance, if the rules are applied to $\ta_2$, all critical letters have to be $\ta_1$, the upper bound is not tight. The following lemma shows how the locality changes if {\em critical letters} are equal.

\input{rulesdifferentborders}

Lemma~\ref{calc} shows two peculiarities: the smaller a letter is w.r.t. the given order the less cases exist in which the locality is changed maximally; words can be categorised w.r.t. their joins and separators - the less of these occurrences appear between different critical letters the smaller is the difference between the locality of the normal form and the word itself. Moreover, the worst case does not incorporate that the worst case for one application of one rule may be the best case for another one such that the increase and decrease cancel each other out. We leave this investigation for general alphabets as an open problem. In the following section, we study the behaviour for alphabets of size up to $3$.

%% file: proof_applR1.tex
\begin{proof}
We are going to prove the claim by comparing the locality resulting from the canonical marking sequence with the locality resulting from an arbitrary marking sequence $\sigma=(y_1,\dots,y_{\ell})$, i.e. the results for the optimal marking sequence may only be better. Consider stage $i\in[\ell]$ and set $I=\{y_1,\dots,y_i\}$.\\
\textbf{case $R_1$:} Let $w=x_1\ta_{k_1}u\ta_{k_2}x_2\ta_{k_3}v\ta_{k_4}x_3$ and 
$R_1(w)=x_1\ta_{k_1}v\ta_{k_2}x_2y_{k_3}u\ta_{k_4}x_3=:w'$. Set $I=\{y_1,\dots,
y_{i}\}$ for a fixed $i\in[n]$. Then after the \nth{$i$} stage of $\sigma'$ we have
\begin{align*}
n(w,I) &= n(x_1\ta_{k_1},I)+n(u,I)+n(\ta_{k_2}x_2\ta_{k_3},I)+n(v,I)+n(\ta_{k_4}x_3,I)\\
&\,-o(\ta_{k_1},u[1],I)-o(u[|u|],\ta_{k_2},I)-o(\ta_{k_3},v[1],I)-o(v[|v|],\ta_{k_4},I)
\end{align*}
and 
\begin{align*}
n(w',I) &= n(x_1\ta_{k_1},I)+n(v,I)+n(\ta_{k_2}x_2\ta_{k_3},I)+n(u,I)+n(\ta_{k_4}x_3,I)\\
&\,-o(\ta_{k_1},v[1],I)-o(v[|v|],\ta_{k_2},I)-o(\ta_{k_3},u[1],I)-o(u[|u|],\ta_{k_4},I)
\end{align*}
Thus we get
\begin{align*}
|n(w',I)&-n(w,I)|\\
= |&-o(\ta_{k_1},u[1],I)-o(u[|u|],\ta_{k_2},I)-o(\ta_{k_3},v[1],I)-o(v[|v|],\ta_{k_4},I)\\
&+o(\ta_{k_1},v[1],I)+o(v[|v|],\ta_{k_2},I)+o(\ta_{k_3},u[1],I)+o(u[|u|],\ta_{k_4},I)|.
\end{align*}
Set $M_1=\{v[1],u[1],u[|u|],v[|v|]\}$, $M_2=\{\ta_{k_1},\ta_{k_2}\,\ta_{k_3},\ta_{k_4}\}$ and 
$M=M_1\cup M_2$. For getting the maximal difference in the locality change, we have to evaluate the different possibilities for $\sigma'$, and thus the different possibilities for $I$.
Consider $J=I\cap M$.
In the case $|J|=0$ the difference is $0$ since non of the summands becomes $1$. Analogously for $|J|=8$ the difference is also $0$ as all summands become 1 and cancel out. For $|J|=1$ also non of the summands can become $1$, so the difference is $0$. For symmetry reasons the difference is also $0$ for $|J|=7$. In the case $|J|=2$ the difference can be either $0$ or $1$. If the two marked letters occur in one summand, the difference is $1$, since all other summands are $0$. If the marked letters do not occur in one summand together, all summands are $0$ and thus the difference is $0$. By symmetry we get an analogous result for $|J|=6$. For $|J|=3$ the difference can be $0$ or $1$. Similar to the case before, only one positive and one negative summand can become $1$ at most. If one negative and one positive summand becomes $1$ the difference is $0$. If no summand becomes $1$, the difference is $0$ as well. But if either one negative or one one positive summand becomes 1 the difference is $1$. Again, for $|J|=5$ for symmetry reasons the same results apply. The case $|J|=4$ is the only one in which the difference can become $0$, $1$ or $2$. Four marked letters can lead to two positive and two negative summands that become $1$, respectively. The largest difference is obtained if either two negative or two positive summands become $1$ and respectively all positive or negative summands are $0$. This leads to a difference of $2$. In all other cases the difference is $0$ or $1$. This concludes the proof for $R_1$.\\
\textbf{case $R_2$:} Consider now
$w=x_1z_1u\overrightarrow{\ta_i}^{k_1}z_2x_2z_3v \overrightarrow{\ta_i}^{k_2}z_4x_3$ and $R_2(w)=x_1z_1u\overrightarrow{\ta_i}^{k_1}$ $\overrightarrow{\ta_i}^{k_2}z_2x_2z_3vz_4x_3=w'$. Then we get
\begin{multline*}
n(w,I) = n(x_1z_1u,I)+n(\overrightarrow{\ta_i}^{k_1},I)+n(z_2x_2z_3v,I)+n(\overrightarrow{\ta_i}^{k_2},I)+n(z_4x_3,I)\\
-o(u[|u|],\ta_i,I)-o(\overrightarrow{\ta_i}^{k_1}[|\overrightarrow{\ta_i}^{k_1}|],z_2,I)-o(v[|v|],\ta_i,I)
-o(\overrightarrow{\ta_i}^{k_2}[|\overrightarrow{\ta_i}^{k_2}|],z_4,I)
\end{multline*}
and 
\begin{multline*}
n(w',I) = n(x_1z_1,u,I)+n(\overrightarrow{\ta_i}^{k_1},I)+n(\overrightarrow{\ta_i}^{k_2},I)+n(z_2x_2z_3v,I)
+n(z_4x_3,I)\\
-o(u[|u|],\ta_i,I)-o(\overrightarrow{\ta_i}^{k_1}|[\overrightarrow{\ta_i}^{k_1}|],\ta_i,I)
-o(\overrightarrow{\ta_i}^{k_2}[|\overrightarrow{\ta_i}^{k_2}|],z_2,I)-o(v[|v|],z_4,I).
\end{multline*}
Thus we get
\begin{multline*}
|n(w',I)-n(w,I)|
=\\
|-o(\overrightarrow{\ta_i}^{k_1}|[\overrightarrow{\ta_i}^{k_1}|],\ta_i,I)
-o(\overrightarrow{\ta_i}^{k_2}[|\overrightarrow{\ta_i}^{k_2}|],z_2,I)-o(v[|v|],z_4,I)\\
+o(\overrightarrow{\ta_i}^{k_1}[|\overrightarrow{\ta_i}^{k_1}|],z_2,I)+o(v[|v|],\ta_i,I)+o(\overrightarrow{\ta_i}^{k_2}[|\overrightarrow{\ta_i}^{k_2}|],z_4,I)|.
\end{multline*}
Set $M_1=\{\overrightarrow{\ta_i}^{k_1}[|\overrightarrow{\ta_i}^{k_1}|], 
\overrightarrow{\ta_i}^{k_2}[|\overrightarrow{\ta_i}^{k_2}|],v[|v|]\}$, $M_2=\{\ta_i,z_2,z_4\}$, and 
$M=M_1\cup M_2$. Analogously to $R_1$, we are distinguishing the possibilities for $\sigma'$. Consider $J=I\cap M$. If $|J|=6$, the difference is obviously 
$0$. If $|J|=5$ the difference is also $0$ since each element from $M$ occurs in exactly one positive and one negative summand. By symmetry we get that the difference is $0$ for $|J|=0$ or $|J|=1$. If $|J|=4$ and either $M_1\cap J$ or $M_2\cap J$ is empty then the difference is $0$ since two negative summands and two positive summands are non-zero. If $J\cap M_1$ and $J\cap M_2$ are non-empty the difference is $1$ since one positive (or negative resp.) summand and two negative (or positive resp.) summands are affected. Again by symmetry we get the analogous result for $|J|=4$. If $|J|=3$, 
the difference is $0$. This concludes the proof for $R_2$.\\
\textbf{case $R_3$:} W.l.o.g. we are only considering $R_3(b)$.
Consider now $w=x_1\overrightarrow{\ta_i}^r c_i \overleftarrow{\ta_i}^s x_2$ according to $R_3$. If $r=0$ the word is not changed by the application and its locality is not either. 
If on the other hand $r>0$ two different cases have to be distinguished in regards to the number of occurrences of satellites on the right side of the core. Firstly, if $s \neq 0$, i.e. satellites occur on both sides, then $R_3(w) = x_1 c_i\overleftarrow{\ta_i}^s(\overrightarrow{\ta_i}^r)^R x_2=:w'$ and we get with $\overrightarrow{\ta_i}[1]=\overleftarrow{\ta_i}[|\overleftarrow{\ta_i}|]=\ta_i$
\begin{align*}
n(w,I) &= n(x_1,I)+n(\overrightarrow{\ta_i}^r,I)+n(c_i\overleftarrow{\ta_i}^s,I)+n(x_2,I)\\
&\,-o(x_1[|x_1|],\ta_i,I)-o(\overrightarrow{\ta_i}^r[|\overrightarrow{\ta_i}^r|],c_i[1],I)-o(\ta_i,x_2[1],I)
\end{align*}
and
\begin{align*}
n(w',I) &= n(x_1,I)+n(c_i\overleftarrow{\ta_i}^s,I)+n((\overrightarrow{\ta_i}^r)^R,I)+n(x_2,I)\\
&\,-o(x_1[|x_1|],c_i[1],I) -o(\ta_i,(\overrightarrow{\ta_i}^r)^R[1],I) -o((\overrightarrow{\ta_i}^r)^R[|(\overrightarrow{\ta_i}^r)^R|],x_2[1],I).
\end{align*}
By $n(\overrightarrow{\ta_i}^r,I)=n((\overrightarrow{\ta_i}^r)^R,I)$ and $(\overrightarrow{\ta_i}^r)^R[|(\overrightarrow{\ta_i}^r)^R|]=\ta_i$ we get 
\begin{multline*}
|n(w',I)-n(w,I)|
=|-o(x_1[|x_1|],c_i[1],I) -o(\ta_i,(\overrightarrow{\ta_i}^r)^R[1],I)\\
+o(x_1[|x_1|],\ta_i,I) +o(\overrightarrow{\ta_i}^r[|\overrightarrow{\ta_i}^r|],c_i[1],I)|,
\end{multline*}
where $(\overrightarrow{\ta_i}^r)^R[1] = \overrightarrow{\ta_i}^r[|\overrightarrow{\ta_i}^r|]$.\\
Set $M_1=\{x_1[|x_1|],(\overrightarrow{\ta_i}^r)^R[1]\}$, $M_2=\{c_i[1],\ta_i\}$, and $M=M_1\cup M_2$. Similarly to $R_1 \mbox{ and } R_2 $, we are distinguishing the possibilities for $\sigma'$. Consider $J=I\cap M$. 
If $|J|=4$, the difference is obviously $0$. If $|J|=3$ the difference is $1$ since each element from $M$ occurs in exactly one positive and one negative summand. By symmetry we get that the difference is $0$ for $|J|=0$ or $|J|=1$. 
If $|J|=2$ and either $M_1\cap J$ or $M_2\cap J$ is empty then the difference is $0$ since two negative summands and two positive summands are non-zero. If $J\cap M_1$ and $J\cap M_2$ are non-empty the difference is $1$ since exactly one positive or negative summand is affected.\\

Secondly, if there are no satellites on the right side of the core ($s=0$) then $w=x_1\overrightarrow{\ta_i}^r c_ix_2\mbox { and } R_3(w) =x_1 c_i (\overrightarrow{\ta_i}^r)^R x_2=:w'$ according to $R_3$.\\
Then we get with $\overrightarrow{\ta_i}^r[1]=(\overrightarrow{\ta_i}^r)^R[|(\overrightarrow{\ta_i}^r)^R|]=a_i$
\begin{align*}
n(w,I) &= n(x_1,I)+n(\overrightarrow{\ta_i}^r,I)+n(c_i,I)+n(x_2,I)\\
&\,-o(x_1[|x_1|],\ta_i,I)-o(\overrightarrow{\ta_i}^r[|\overrightarrow{\ta_i}^r|],c_i[1],I)-o(c_i[|c_i|],x_2[1],I)
\end{align*}
and
\begin{align*}
n(w',I) &= n(x_1,I)+n(c_i,I)+n((\overrightarrow{\ta_i}^r)^R,I)+n(x_2,I)\\
&\,-o(x_1[|x_1|],c_i[1],I) -o(c_i[|c_i|],(\overrightarrow{\ta_i}^r)^R[1],I) -o(a_i,x_2[1],I).
\end{align*}
By $n(\overrightarrow{\ta_i}^r,I)=n((\overrightarrow{\ta_i}^r)^R,I)$ we get 
\begin{multline*}
|n(w',I)-n(w,I)|
=|-o(x_1[|x_1|],c_i[1],I) -o(c_i[|c_i|],(\overrightarrow{\ta_i}^r)^R[1],I) -o(a_i,x_2[1],I) \\
+o(x_1[|x_1|],\ta_i,I) +o(\overrightarrow{\ta_i}^r[|\overrightarrow{\ta_i}^r|],c_i[1],I) +o(c_i[|c_i|],x_2[1],I)|.
\end{multline*}

Consider $M_1=\{x_1[|x_1|],\overrightarrow{\ta_i}^r[|\overrightarrow{\ta_i}^r|], x_2[1]\}$, $M_2=\{c_i[1], c_i[|c_i|], \ta_i\}$, and $M=M_1\cup M_2$. Then we get analogously to case $R_2$ a difference of at most 1. 
The proof for the application of $R_3(a)$ follows symmetrically. Thus, the absolute value of difference in locality for $w \mbox{ and } w'$ is at most $1$ for the case $R_3$.\qed
\end{proof}

%% file: proof_multappl.tex
\begin{proof}
In the worst case all letters are different and get an increase of $2$ for each application of $R_1$ as well as an increase of $1$ for each application of $R_2$. Interestingly, the locality does not increase with the number of satellites but increases by $1$ at most for each letter.

Since we only consider neighbourless words at any stage $1 < i < l$ in the marking process $w_i$ is of the form $\overrightarrow{\ta_l}^{k_l} \cdots \overrightarrow{\ta_i}^{k_i} c_i \overleftarrow{\ta_i}^{k'_i} \cdots \overleftarrow{\ta_l}^{k'_l} \mbox{ with } k_j, k'_j \in\N_0, i \leq j \leq l$ according to the notation in Remark~\ref{arrowNotation} and \cite{FSTTCS}. To bring $w$ into normal form we apply $R_3(b)$ once in every such stage, moving the factor $\overrightarrow{\ta_i}^{k_i}$ which contains all left satellites of $\ta_i$ to the right side of the core. The order in which these possibly different $\overrightarrow{\ta_i}$ occur there is not of importance since they are of the form $\ta_ix \mbox{ with } x\in\{\ta_{i+1}, \cdots, \ta_{\ell}\}^+$ and all letters of $x$ are either join or singleton occurrences greater than $\ta_i$ and are moved with $R_1$ or $R_2$ in the remaining marking steps, if necessary. For both $\ta_1$ and $\ta_l$ there is no application of $R_3$ needed since no occurrence of $\ta_1$ has to be moved and there are no satellites for $\ta_l$ (which are all joining occurrences).
\qed
\end{proof}

%% file: rulesdifferentborders.tex
\ifpaper
\begin{lemma}[$\ast$]
\else
\begin{lemma}
\fi
\label{calc}
Let $w\in\Sigma^{\ast}$. 
Regarding $R_1$ we have that the locality does not change if the critical letters are identical and it changes by at most $1$ if three critical letters are equal and the fourth is different or if $\ta_1=\ta_3$ or $\ta_2=\ta_4$. Regarding $R_2$ the results are similar: if the critical letters $\overrightarrow{\ta_i}^{k_1}[|\overrightarrow{\ta_i}^{k_1}|],z_2,z_4$, and $v[|v|]$ are all equal or if $\overrightarrow{\ta_i}^{k_1}[|\overrightarrow{\ta_i}^{k_1}|]$ and $z_2=z_4$ the locality does not change.  Finally regarding $R_3(b)$ the locality does not change if both $x_1[|x_1|] =  x_2[1] \mbox{ and } c_i[1] = c_i \overleftarrow{\ta_i}^s[|c_i \overleftarrow{\ta_i}^s|]$ (including the case $x_1 = x_2 = \varepsilon$).
\end{lemma}
\ifpaper
\else \input{proof_calc}

\fi

%% file: proof_calc.tex
\begin{proof}
Consider firstly $R_1$. For $\ta_{k_1} = \ta_{k_2}= \ta_{k_3}= \ta_{k_4}$ we have
\begin{multline*}
|n(w',I)-n(w,I)|\\
= |-o(\ta_{k_1},u[1],I)-o(u[|u|],\ta_{k_1},I)-o(\ta_{k_1},v[1],I)
-o(v[|v|],\ta_{k_1},I) \\
+o(\ta_{k_1},v[1],I)+o(v[|v|],\ta_{k_1},I)+o(\ta_{k_1},u[1],I)+o(u[|u|],\ta_{k_1},I)|
= 0.
\end{multline*}
If we have $\ta_{k_1}= \ta_{k_3}, \ta_{k_2} = \ta_{k_4}$, we also get
\begin{multline*}
|n(w',I)-n(w,I)| \\
=|-o(\ta_{k_1},u[1],I)-o(u[|u|],\ta_{k_2},I)-o(\ta_{k_1},v[1],I)-o(v[|v|],\ta_{k_2},I)\\
+o(\ta_{k_1},v[1],I)+o(v[|v|],\ta_{k_2},I)+o(\ta_{k_1},u[1],I)+o(u[|u|],\ta_{k_2},I)|
= 0.
\end{multline*}
Consider now $\ta_{k_2}= \ta_{k_3}= \ta_{k_4}$ (the other cases where three critical letters are equal but not the fourth are analogous). Then we get
\begin{multline*}
|n(w',I)-n(w,I)|\\
= |-o(\ta_{k_1},u[1],I)-o(\ta_{k_2},v[1],I)+o(\ta_{k_1},v[1],I)+o(\ta_{k_2},u[1],I)|
\leq 1.
\end{multline*}
Moreover we have for $\ta_{1}=\ta_3$ ($\ta_2=\ta_4$ is analogous)
%
%
%
\begin{multline*}
|n(w',I)-n(w,I)|\\
= |-o(u[|u|],\ta_{k_2},I)-o(v[|v|],\ta_{k_4},I)+o(v[|v|],\ta_{k_2},I)+o(u[|u|],\ta_{k_4},I)|
\leq 1.
\end{multline*}
For $\ta_{k_1}= \ta_{k_2}, \ta_{k_3} = \ta_{k_4}$ we have
\begin{multline*}
|n(w',I)-n(w,I)|\\
= |-o(\ta_{k_1},u[1],I)-o(u[|u|],\ta_{k_1},I)-o(\ta_{k_3},v[1],I)-o(v[|v|],\ta_{k_3},I)\\
+o(\ta_{k_1},v[1],I)+o(v[|v|],\ta_{k_1},I)+o(\ta_{k_3},u[1],I)+o(u[|u|],\ta_{k_3},I)|
\leq 2.
\end{multline*}
For $\ta_{k_1}= \ta_{k_4}, \ta_{k_2} = \ta_{k_3}$ we have
\begin{multline*}
|n(w',I)-n(w,I)|\\
= |-o(\ta_{k_1},u[1],I)-o(u[|u|],\ta_{k_2},I)-o(\ta_{k_2},v[1],I)-o(v[|v|],\ta_{k_1},I)\\
+o(\ta_{k_1},v[1],I)+o(v[|v|],\ta_{k_2},I)+o(\ta_{k_2},u[1],I)+o(u[|u|],\ta_{k_1},I)|
\leq 2.
\end{multline*}
For $\ta_{k_1}= \ta_{k_2}$ we have
\begin{multline*}
|n(w',I)-n(w,I)|\\
= |-o(\ta_{k_1},u[1],I)-o(u[|u|],\ta_{k_1},I)-o(\ta_{k_3},v[1],I)-o(v[|v|],\ta_{k_4},I)\\
+o(\ta_{k_1},v[1],I)+o(v[|v|],\ta_{k_1},I)+o(\ta_{k_3},u[1],I)+o(u[|u|],\ta_{k_4},I)|
\leq 2.
\end{multline*}
For $\ta_{k_1}= \ta_{k_4}$ we have
\begin{multline*}
|n(w',I)-n(w,I)|\\
= |-o(\ta_{k_1},u[1],I)-o(u[|u|],\ta_{k_2},I)-o(\ta_{k_3},v[1],I)-o(v[|v|],\ta_{k_1},I)\\
+o(\ta_{k_1},v[1],I)+o(v[|v|],\ta_{k_2},I)+o(\ta_{k_3},u[1],I)+o(u[|u|],\ta_{k_1},I)|
\leq 2.
\end{multline*}
For $\ta_{k_2}= \ta_{k_3}$ we have
\begin{multline*}
|n(w',I)-n(w,I)|\\
= |-o(\ta_{k_1},u[1],I)-o(u[|u|],\ta_{k_2},I)-o(\ta_{k_2},v[1],I)-o(v[|v|],\ta_{k_4},I)\\
+o(\ta_{k_1},v[1],I)+o(v[|v|],\ta_{k_2},I)+o(\ta_{k_2},u[1],I)+o(u[|u|],\ta_{k_4},I)|
\leq 2.
\end{multline*}
For $\ta_{k_3}= \ta_{k_4}$ we have as a last case for $R_1$
\begin{multline*}
|n(w',I)-n(w,I)|\\
= |-o(\ta_{k_1},u[1],I)-o(u[|u|],\ta_{k_2},I)-o(\ta_{k_3},v[1],I)-o(v[|v|],\ta_{k_3},I)\\
+o(\ta_{k_1},v[1],I)+o(v[|v|],\ta_{k_2},I)+o(\ta_{k_3},u[1],I)+o(u[|u|],\ta_{k_3},I)|
\leq 2.
\end{multline*}
This proves the claim for $R_1$. Consider now $R_2$. If $\overrightarrow{\ta_i}^{k_1}[|\overrightarrow{\ta_i}^{k_1}|] = z_2 = z_4 = v[|v|]$
we get 
\begin{multline*}
|n(w',I)-n(w,I)| \\
= |-o(z_2,\ta_i,I)
-o(\overrightarrow{\ta_i}^{k_2}[|\overrightarrow{\ta_i}^{k_2}|],z_2,I)
-o(z_2,z_2,I)\\
+o(z_2,z_2,I)+o(z_2,\ta_i,I)+o(\overrightarrow{\ta_i}^{k_2}[|\overrightarrow{\ta_i}^{k_2}|],z_2,I)|
= 0.
\end{multline*}
If $\overrightarrow{\ta_i}^{k_1}[|\overrightarrow{\ta_i}^{k_1}|] = v[|v|], z_2 = z_4 $ we get
\begin{multline*}
|n(w',I)-n(w,I)|
= |-o(v|[v|],\ta_i,I)-o(\overrightarrow{\ta_i}^{k_2}[|\overrightarrow{\ta_i}^{k_2}|],z_2,I)
-o(v[|v|],z_2,I)\\
+o(v[|v|],z_2,I)+o(v[|v|],\ta_i,I)+o(\overrightarrow{\ta_i}^{k_2}[|\overrightarrow{\ta_i}^{k_2}|],z_2,I)|
= 0.
\end{multline*}
In the remaining cases the locality may change by at most $1$. If  $z_2 = z_4 = v[|v|]$ (the other cases if three critical letters are equal but not the fourth are analogous) we get
\begin{multline*}
|n(w',I)-n(w,I)|
= |-o(\overrightarrow{\ta_i}^{k_1}[|\overrightarrow{\ta_i}^{k_1}|],\ta_i,I)-o(z_2,z_2,I)\\
+o(\overrightarrow{\ta_i}^{k_1}[|\overrightarrow{\ta_i}^{k_1}|],z_2,I)+o(z_2,\ta_i,I)|
\leq 1.
\end{multline*}
If $\overrightarrow{\ta_i}^{k_1}[|\overrightarrow{\ta_i}^{k_1}|] = z_2, z_4 = v[|v|]$ we get
\begin{multline*}
|n(w',I)-n(w,I)|\\
= |-o(\overrightarrow{\ta_i}^{k_1}[|\overrightarrow{\ta_i}^{k_1}|],\ta_i,I)-o(\overrightarrow{\ta_i}^{k_2}[|\overrightarrow{\ta_i}^{k_2}|],\overrightarrow{\ta_i}^{k_1}[|\overrightarrow{\ta_i}^{k_1}|],I)-o(z_4,z_4,I)\\
+o(\overrightarrow{\ta_i}^{k_1}[|\overrightarrow{\ta_i}^{k_1}|],\overrightarrow{\ta_i}^{k_1}[|\overrightarrow{\ta_i}^{k_1}|],I)+o(z_4,\ta_i,I)+o(\overrightarrow{\ta_i}^{k_2}[|\overrightarrow{\ta_i}^{k_2}|],z_4,I)|
\leq 1.
\end{multline*}
If $\overrightarrow{\ta_i}^{k_1}[|\overrightarrow{\ta_i}^{k_1}|] = z_4, z_2 = v[|v|]$ we get
\begin{multline*}
|n(w',I)-n(w,I)|\\
= |-o(\overrightarrow{\ta_i}^{k_1}[|\overrightarrow{\ta_i}^{k_1}|],\ta_i,I)-o(\overrightarrow{\ta_i}^{k_2}[|\overrightarrow{\ta_i}^{k_2}|],z_2,I)\\
+o(z_2,\ta_i,I)+o(\overrightarrow{\ta_i}^{k_2}[|\overrightarrow{\ta_i}^{k_2}|],\overrightarrow{\ta_i}^{k_1}[|\overrightarrow{\ta_i}^{k_1}|],I)|
\leq 1.
\end{multline*}
If $\overrightarrow{\ta_i}^{k_1}[|\overrightarrow{\ta_i}^{k_1}|] = z_2$ we get
\begin{multline*}
|n(w',I)-n(w,I)|\\
= |-o(\overrightarrow{\ta_i}^{k_1}[|\overrightarrow{\ta_i}^{k_1}|],\ta_i,I)-o(\overrightarrow{\ta_i}^{k_2}[|\overrightarrow{\ta_i}^{k_2}|],\overrightarrow{\ta_i}^{k_1}[|\overrightarrow{\ta_i}^{k_1}|],I)-o(v[|v|],z_4,I)\\
+o(\overrightarrow{\ta_i}^{k_1}[|\overrightarrow{\ta_i}^{k_1}|],\overrightarrow{\ta_i}^{k_1}[|\overrightarrow{\ta_i}^{k_1}|],I)+o(v[|v|],\ta_i,I)+o(\overrightarrow{\ta_i}^{k_2}[|\overrightarrow{\ta_i}^{k_2}|],z_4,I)|
\leq 1.
\end{multline*}
If $\overrightarrow{\ta_i}^{k_1}[|\overrightarrow{\ta_i}^{k_1}|] = z_4$ we get
\begin{multline*}
|n(w',I)-n(w,I)|\\
= |-o(\overrightarrow{\ta_i}^{k_1}[|\overrightarrow{\ta_i}^{k_1}|],\ta_i,I)-o(\overrightarrow{\ta_i}^{k_2}[|\overrightarrow{\ta_i}^{k_2}|],z_2,I)-o(v[|v|],\overrightarrow{\ta_i}^{k_1}[|\overrightarrow{\ta_i}^{k_1}|],I)\\
+o(\overrightarrow{\ta_i}^{k_1}[|\overrightarrow{\ta_i}^{k_1}|],z_2,I)+o(v[|v|],\ta_i,I)+o(\overrightarrow{\ta_i}^{k_2}[|\overrightarrow{\ta_i}^{k_2}|],\overrightarrow{\ta_i}^{k_1}[|\overrightarrow{\ta_i}^{k_1}|],I)|
\leq 1.
\end{multline*}
If $\overrightarrow{\ta_i}^{k_1}[|\overrightarrow{\ta_i}^{k_1}|] = v[|v|]$ w get
\begin{multline*}
|n(w',I)-n(w,I)| 
= |-o(\overrightarrow{\ta_i}^{k_2}[|\overrightarrow{\ta_i}^{k_2}|],z_2,I)-o(\overrightarrow{\ta_i}^{k_1}[|\overrightarrow{\ta_i}^{k_1}|],z_4,I)\\
+o(\overrightarrow{\ta_i}^{k_1}[|\overrightarrow{\ta_i}^{k_1}|],z_2,I)+o(\overrightarrow{\ta_i}^{k_2}[|\overrightarrow{\ta_i}^{k_2}|],z_4,I)|
\leq 1,
\end{multline*}
If $z_2 = z_4$ we get
\begin{multline*}
|n(w',I)-n(w,I)|
= |-o(\overrightarrow{\ta_i}^{k_1}[|\overrightarrow{\ta_i}^{k_1}|],\ta_i,I)-o(v[|v|],z_2,I)\\
+o(\overrightarrow{\ta_i}^{k_1}[|\overrightarrow{\ta_i}^{k_1}|],z_2,I)+o(v[|v|],\ta_i,I)|\leq 1.
\end{multline*}
If  $z_2 = v[|v|]$ we get
\begin{multline*}
|n(w',I)-n(w,I)|
= |-o(\overrightarrow{\ta_i}^{k_1}[|\overrightarrow{\ta_i}^{k_1}|],\ta_i,I)-o(\overrightarrow{\ta_i}^{k_2}[|\overrightarrow{\ta_i}^{k_2}|],z_2,I)-o(z_2,z_4,I)\\
+o(\overrightarrow{\ta_i}^{k_1}[|\overrightarrow{\ta_i}^{k_1}|],z_2,I)+o(z_2,\ta_i,I)+o(\overrightarrow{\ta_i}^{k_2}[|\overrightarrow{\ta_i}^{k_2}|],z_4,I)|\leq 1.
\end{multline*}
Finally for $R_2$ if $z_4 = v[|v|]$ we get
\begin{multline*}
|n(w',I)-n(w,I)|
= |-o(\overrightarrow{\ta_i}^{k_1}[|\overrightarrow{\ta_i}^{k_1}|],\ta_i,I)-o(\overrightarrow{\ta_i}^{k_2}[|\overrightarrow{\ta_i}^{k_2}|],z_2,I)-o(z_4,z_4,I)\\
+o(\overrightarrow{\ta_i}^{k_1}[|\overrightarrow{\ta_i}^{k_1}|],z_2,I)+o(z_4,\ta_i,I)+o(\overrightarrow{\ta_i}^{k_2}[|\overrightarrow{\ta_i}^{k_2}|],z_4,I)|\leq 1.
\end{multline*}

This concludes the proof for $R_2$. Considering $R_3(b)$ there are less options for equality of the critical letters $x_1[|x_1|], c_i[1], c_i \overleftarrow{\ta_i}^s[|c_i \overleftarrow{\ta_i}^s|] \mbox{ and } x_2[1]$. By the construction of the neighbourless word $w_i$ to which $R_3$ is applied we know that $w_i=x_1\overrightarrow{\ta_i}^r c_i \overleftarrow{\ta_i}^s x_2$ with $r > 0$ and $c_i$ the core of $w_i$. According to $R_3$ and the definition of the core $c_i$, we know that $x_1[|x_1|], x_2[1] >_{\sigma_\Sigma} \ta_i \geq_{\sigma_\Sigma} c_i[1], c_i \overleftarrow{\ta_i}^s[|c_i \overleftarrow{\ta_i}^s|]$ and specifically $x_1[|x_1|], x_2[1] \neq c_i[1], c_i \overleftarrow{\ta_i}^s[|c_i \overleftarrow{\ta_i}^s|]$ and the factors $c_i[1] \mbox{ and } c_i \overleftarrow{\ta_i}^s[|c_i \overleftarrow{\ta_i}^s|]$ are not the empty word, whereas $x_1 \mbox{ and } x_2$ might be. Consider first $x_1 \mbox{ and } x_2$ not empty.
\\

If $x_1[|x_1|] =  x_2[1] \mbox{ and } c_i[1] = c_i \overleftarrow{\ta_i}^s[|c_i \overleftarrow{\ta_i}^s|]$ we get
\begin{multline*}
|n(w',I)-n(w,I)| = \\
|-o(x_1[|x_1|],c_i[1],I) -o(c_i[1],(\overrightarrow{\ta_i}^r)^R[1],I) -o(\ta_i,x_1[|x_1|],I) \\
+o(x_1[|x_1|],\ta_i,I) +o(\overrightarrow{\ta_i}^r[|\overrightarrow{\ta_i}^r|],c_i[1],I) +o(c_i[1],x_1[|x_1|],I)| = 0.
\end{multline*}

If $x_1[|x_1|] =  x_2[1] \mbox{ and } c_i[1] \neq c_i \overleftarrow{\ta_i}^s[|c_i \overleftarrow{\ta_i}^s|]$ (the case for $x_1[|x_1|] \neq  x_2[1]$
 and $c_i[1] = c_i \overleftarrow{\ta_i}^s[|c_i \overleftarrow{\ta_i}^s|]$ follows analogously) we get
\begin{multline*}
|n(w',I)-n(w,I)| = \\
|-o(x_1[|x_1|],c_i[1],I) 
-o(c_i \overleftarrow{\ta_i}^s[|c_i \overleftarrow{\ta_i}^s|],(\overrightarrow{\ta_i}^r)^R[1],I) 
-o(\ta_i,x_1[|x_1|],I) \\
+o(x_1[|x_1|],\ta_i,I) 
+o(\overrightarrow{\ta_i}^r[|\overrightarrow{\ta_i}^r|],c_i[1],I) 
+o(c_i \overleftarrow{\ta_i}^s[|c_i \overleftarrow{\ta_i}^s|],x_2[1],I)| \leq 1.
\end{multline*}

If $x_1[|x_1|] \neq  x_2[1] \mbox{ and } c_i[1] \neq c_i \overleftarrow{\ta_i}^s[|c_i \overleftarrow{\ta_i}^s|]$ we get
\begin{multline*}
|n(w',I)-n(w,I)| = \\
|-o(x_1[|x_1|],c_i[1],I) -o(c_i \overleftarrow{\ta_i}^s[|c_i \overleftarrow{\ta_i}^s|],\ta_i,I) -o(\ta_i,x_2[1],I) \\
+o(x_1[|x_1|],\ta_i,I) +o(\overrightarrow{\ta_i}^r[|\overrightarrow{\ta_i}^r|],c_i[1],I) +o(c_i \overleftarrow{\ta_i}^s[|c_i \overleftarrow{\ta_i}^s|],x_2[1],I)| \leq 1.
\end{multline*}

If on the other hand $x_1 \mbox{ or } x_2$ are empty there are less critical letters. The previously given calculations include these cases since the summands in which either $x_1[|x_1|]$ or $x_2[1]$ appears are $0$ by definition of $o$ and do not influence the inequalities. 
\qed
\end{proof}

%% file: card3.tex
In this section, we are using $\ta,\tb$, and $\tc$ for the alphabet for better readability. For unary alphabets we have exactly one word containing of a single letter since we only consider condensed words.
The binary case $\Sigma=\{\ta,\tb\}$ can also shortly be explained: blocksequences are of the form 
$(b_1,1)$. Again, since the words are condensed and neighbourless, each word has to be an alternation of $\ta$ and $\tb$ and assuming $\sigma_{\Sigma}$ the word starts and ends with $\ta$. Thus we have $b_1$ occurrences of $\ta$ and $b_1-1$ join occurrences of $\tb$.  This leads immediately to the fact that the only other marking sequence is better (and thus optimal) since we obtain the blocksequence $(b_1-1,1)$. 
In the case $\Sigma=\{\ta,\tb,\tc\}$ $\ebs$ are of the form $\gamma=((b_1,b_2,1),j_1,s_1)$ (omitting some brackets for better readability)
implying $w_{\gamma} =(\ta\tb)^{j_1}\ta(\tc\tb)^{s_1}(\tc\ta)^{b_1-j_1-1}(\tc\tb)^{b_2-b_1-s_1+j_1}$.
Firstly, we show how the locality of $w$ and $w_{\gamma}$ differ. Notice that in the case $|\Sigma|=3$ only occurrences of $\tb$ may be join- or separating occurrences and all occurrences of $\tc$ are joins.

\ifpaper
\begin{proposition}[$\ast$]
\else
\begin{proposition}
\fi
\label{locchanges3}
Let $w\in\Sigma^{\ast}$ and $\gamma=((b_1,b_2,1),j_1,s_1)$ the $\ebs$ while marking with $\sigma_{\Sigma}$. Then we have $\loc_{\sigma_{\Sigma}}(w)=\loc_{\sigma_{\Sigma}}(w_{\gamma})$.
\end{proposition}

\ifpaper
\else \input{proof_locchanges3}
\fi

Thus, on a ternary alphabet we may assume the normal form without any restriction w.r.t. $\loc_{\sigma_{\Sigma}}$. The following proposition determines the optimal marking sequence for the normal form just by the $\ebs$.

\ifpaper
\begin{theorem}[$\ast$]
\else
\begin{theorem}
\fi
\label{optcard3}
Given a valid $\ebs$ $\gamma=((b_1,b_2,1),j_1,s_1)$  and $w_{\gamma}$ w.r.t. $\sigma_{\Sigma}$  the optimal marking sequence is given by\\
- $(\tb,\tc,\ta)$ if $2b_1\geq2j_1+b_2$ and $b_2-1\geq b_1$ or $2b_1\leq2j_1+b_2$ and $b_1\geq2j_1+1$,\\
- $(\tc,\ta,\tb)$ if $b_1\leq 2j_1+1$ and $2b_1\geq2j_1+b_2$ or $b_1\geq2j_1+1$ and $b_1\geq b_2-1$,\\
- $(\tc,\tb,\ta)$ if $b_1\geq b_2-1$ and $2b_1\leq2j_1+b_2$ or $2b_1\geq2j_1+b_2$ and $b_1\leq 2j_1+1$,\\
- $(\ta,\tb,\tc)$ otherwise. 	
\end{theorem}

\ifpaper
\else
\input{proof_optcard3}
\fi

Thus, in the ternary case we are able to determine the optimal marking sequence for a neighbourless word with a constant number of arithmetic operations and comparisons if the extended marking sequence is given; notice that the normal form does not have to be computed since only the information from the extended blocksequence is needed.

%% file: proof_locchanges3.tex
\begin{proof}
Given $j_1$ we know that rule $R_1$ has to be applied at most $j_1+1$ times, namely $j_1$ for bringing the join occurrences at the correct position and $1$ for bringing 
one gap with separating occurrences of $\tb$ at the correct position (in the case that some are already in the correct position, the number of applications decreases).
Moreover, we may assume for the input of $R_1$: $\ta_{k_j}=\ta$ for $j\in[4]$, $u[1]=u[|u|]=\tc$ since $u$ needs to be an $\ta$-gap which does not contain a join-occurrence of $\tb$ (otherwise we do not apply $R_1$), and $v[1]=v[|v|]=\tb$. Thus, we get for $I\subseteq\Sigma^{\ast}$
\begin{multline*}
|n(w',I)-n(w,I)|
= |-o(\ta,\tc,I)-o(\tc,\ta,I)-o(\ta,\tb,I)-o(\tb,\ta,I)\\
+o(\ta,\tb,I)+o(\tb,\ta,I)+o(\ta,\tc,I)+o(\tc,\ta,I)|=0,
\end{multline*}
i.e. moving the join occurrences to the correct positions does not change the locality at all. In the next step, we have to move one {\em separating gap} to the left, i.e. we have to exchange a join occurrence of $\tc$ with an occurrence of the form $(\tc\tb)^{\ell}\tc$. In this case we have again $\ta_{k_j}=\ta$, $u[1]=u[|u|]=\tc$ and $v[1]=v[|v|]=\tc$ 
resulting in
\begin{multline*}
|n(w',I)-n(w,I)|
= |-o(\ta,\tc,I)-o(\tc,\ta,I)-o(\ta,\tc,I)-o(\tc,\ta,I)\\
+o(\ta,\tc,I)+o(\tc,\ta,I)+o(\ta,\tc,I)+o(\tc,\ta,I)|=0
\end{multline*}
and hence, applying $R_1$ never changes the locality. In the next step we are going to move all separating occurrences into the same gap (the one where we just put one such block by $R_1$). Here we know 
$\overrightarrow{\ta_i}^{k_j}[|\overrightarrow{\ta_i}^{k_j}|]=\tc$ for $j\in[2]$, $\ta_i=\tb$, $v[|v|]=c$, and $z_j=\ta$ for $j\in[4]$ and we get
\begin{multline*}
|n(w',I)-n(w,I)| = |-o(\tc,\tb,I)-o(\tc,\ta,I)-o(\tc,\ta,I)\\
+o(\tc,\ta,I)+o(\tc,\tb,I)+o(\tc,\ta,I)|=0.
\end{multline*}
Hence, the application of $R_2$ does not change the locality either. Finally, we look at the application of $R_3$ and notice $\overleftarrow{\ta_i}^R[|\overleftarrow{\ta_i}^R|]=\tc$, $\ta_i=\tb$, $c_i[|c_i|]=\ta$, $x_2[1]=\tc$, $\overleftarrow{\ta_i}[1]=\tc$ and obtain
\begin{multline*}
|n(w',I)-n(w,I)|
=|-o(\tc,\tb,I)
-o(\ta,\tc,I)+o(\ta,\tc,I)+o(\tb,\tc,I)|=0
\end{multline*}
which leads to $\loc_{\sigma_{\Sigma}}(w)=\loc_{\sigma_{\Sigma}}(w_{\gamma})$. \qed
\end{proof}

%% file: proof_optcard3.tex
\begin{proof}
Notice that for determining the locality of a word, it suffices to calculate all blocksequences of the possible marking sequence; the extended blocksequence is not of interest.
We will here only calculate the blocksequence for one marking sequence in detail since the calculation is similar in all cases (the results are depicted in Table \ref{stupidtable}. Set for convenience $u_1=(\ta\tb)^{j_1}\ta$, $u_2=(\tc\tb)^{s_1}$, $u_3=(\tc\ta)^{b_1-j_1-1}$, and $u_4=(\tc\tb)^{b_2-b_1-s_1+j_1}$. Consider the marking sequence $(\tb,\tc,\ta)$. Marking $\tb$ leads to $j_1$ marked blocks in $u_1$ (the last letter is unmarked), $s_1$ marked blocks in $u_2$ (the last letter is marked), no marked block in $u_3$, and $b_2-b_1-s_1+j_1$ marked blocks in $u_4$. This leads
to $j_1+s_1+b_2-b_1-s_1+j_1=2j_1+b_2-b_1$ marked blocks. Now marking $\tc$ leads to $j_1$ marked blocks in $u_1$ (the last letter is still unmarked), $u_2$ is one marked block, $b_1-j_1-1$ marked blocks in $u_3$ (the first letter is marked, the last letter is unmarked), and $u_4$ is one marked block. Thus we get $j_1+1+b_1-j_1-2+1=b_1$.

\begin{table}[]
\centering
\begin{tabular}{|l|l|l|}
\hline
marking sequence & blocks after 1$^{\mbox{\tiny st}}$ marked letter & blocks after 2$^{\mbox{\tiny nd}}$ marked letter \\
\hline
$\sigma_1=(\ta,\tb,\tc)$ & $b_1$  & $b_2$                             \\
$\sigma_2=(\ta,\tc,\tb)$ & $b_1$                            & $2j_1+b_2-b_1$                    \\
$\sigma_3=(\tb,\ta,\tc)$ & $2j_1+b_2-b_1$                   & $b_2$                             \\
$\sigma_4=(\tb,\tc,\ta)$ & $2j_1+b_2-b_1$                   & $b_1$                             \\
$\sigma_5=(\tc,\ta,\tb)$ & $b_2-1$                          & $2j_1+b_2-b_1$                    \\
$\sigma_6=(\tc,\tb,\ta)$ & $b_2-1$                          & $b_1$\\
\hline                            
\end{tabular}
\caption{Number of blocks after marking the first and after marking the second letter for each marking sequence.}
\label{stupidtable}
\end{table}
This information can now be used to derive which sequences are optimal. Recall that a marking sequence is optimal if there is no other marking sequence that leads to a smaller locality.
We know that $\loc_{(\ta,\tb,\tc)}=\min\{b_1,b_2\}$ and thus any other marking sequence is only better if it needs at most $\loc_{(\ta,\tb)}$ blocks while marking. Notice from Table \ref{stupidtable} that $\sigma_2$ and $\sigma_4$ result in the same locality - therefore we are only considering $\sigma_2$.\\
\textbf{case 1:} $b_1\leq b_2$\\
In this case $\sigma_2,\sigma_3$ and $\sigma_5$ are better than $\sigma_1$ if $2j_1+b_2-b_1<b_2$ 
holds. This is equivalent to $2j_1<b_1$; $\sigma_6$ is in any case better than $\sigma_1$. If both
conditions are true, $\sigma_6$ is worse than the other ones since the opposite led to $b_2-1<2j_1+b_2-b_1$ which is a contradiction to $b_1<2j_1$.\\
\textbf{case 2:} $b_1>b_2$\\
In this case $\sigma_2,\sigma_3$ and $\sigma_5$ are better than $\sigma_1$ if $2j_1+b_2-b_1<b_1$ holds. This is equivalent to $2j_1+b_2<2b_1$. The last marking sequence $\sigma_6$ is never better than $\sigma_1$.\qed
\end{proof}

%% file: conclusion.tex
In this paper, we investigated a new point of view regarding the notion of $k$-locality. While previous works
were focussed on the locality of one single word and the connection to other domains (especially pattern matching or graph theory), we 
introduced the notion of blocksequence for grouping words and finding similarities of these words. We noticed that 
just a blocksequence does not provide enough information for a reasonable characterisation, since too many words with
different locality fall into the same class. Thus, we strengthened this notion, and introduced extended blocksequences. These sequences not only
count the number of marked blocks, in each step of a marking sequence, but also provide information about the roles of single letters: neighbours, joins, 
separators, and satellites. Further, we focused our analysis on neighbourless words. In that case, we were able to define 
a normal form for each class, and compute it in linear time. We have also shown an upper bound on the difference between the locality of a word and 
that of its normal form. It remains open to determine the exact difference between these two for a specific word over an alphabet with at least four letters. 
We conjecture that our upper bound is actually not tight, since the worst case for one of the applied rules can be cancelled out with the application of the
next rule. Surprisingly for us, a computer programme showed that the locality of a word and that of its normal form, over a six-letter alphabet, differ by at most seven, 
independent of the number of satellites, joins, and separators. For a three letter alphabet we gave a full characterisation
including the optimal marking sequence of a word, as determined by the extended blocksequence.

In this work, we merely started the study of this new perspective on the locality of words. Further problems, such as the computation of the normal form's locality and a deeper understanding
of the locality changes between a word and its normal form, are left as future work.

%% file: appendix.tex
\noindent
\textbf{Proof of Theorem~\ref{ntupel}.}

\input{proof_ntupel}
\bigskip

\noindent
\textbf{A Characterisation of all condensed words of $\mathfrak{W}_\beta$ (including a characterisation of the shortest words from this set)}

\input{proof_condensed}
\bigskip

\noindent
\textbf{Proof of Theorem~\ref{validity}.}

\input{proof_validity}
\bigskip

\noindent
\textbf{Proof of Theorem~\ref{lengthforebs}.}

\input{proof_lengthforebs}
\bigskip

\noindent
\textbf{Proof of Lemma~\ref{occletter}.}

\input{proof_occletter}
\bigskip

\noindent
\textbf{Proof of Theorem~\ref{orderofletters}.}

\input{proof_orderofletters}
\newpage

\noindent
\textbf{Proof of Proposition~\ref{nlcheck}.}

\input{proof_nlcheck}
\bigskip

\noindent
\textbf{Proof of Proposition~\ref{ebs}.}

\input{proof_ebs}
\bigskip

\noindent
\textbf{Proof of Theorem~\ref{obtainingnormalform}.}

\input{proof_obtainingnormalform}
\bigskip

\noindent
\textbf{Proof of Lemma~\ref{class}.}

\input{proof_class}
\bigskip

\noindent
\textbf{Proof of Corollary~\ref{charVgamma}.}

\input{proof_charVgamma}
\bigskip

\noindent
\textbf{Proof of Theorem~\ref{applR1}.}

\input{proof_applR1}
\bigskip

\noindent
\textbf{Proof of Corollary~\ref{multappl}.}

\input{proof_multappl}
\bigskip

\noindent
\textbf{Proof of Lemma~\ref{calc}.}

\input{proof_calc}
\bigskip

\noindent
\textbf{Proof of Proposition~\ref{locchanges3}.}

\input{proof_locchanges3}
\newpage

\noindent
\textbf{Proof of Theorem~\ref{optcard3}.}

\input{proof_optcard3}